\documentclass[
 prapplied,twocolumn,
 amsmath,amssymb,
 aps,nofootinbib,longbibliography,
]{revtex4-2}

\usepackage{amsthm}
\usepackage{graphicx}
\usepackage{color}
\usepackage{dcolumn}
\usepackage{bm}
\usepackage{mathtools}      % \coloneqq
\usepackage{mathrsfs}       % for operator symbol
\usepackage{braket}
\usepackage{float}
\usepackage[table]{xcolor}

\newtheorem{lemma}{Lemma}

\newtheorem{corollary}{Corollary}

\newcommand{\cA}{{\mathcal{A}}}

\newcommand{\cD}{{\mathcal{D}}}
\newcommand{\cE}{{\mathcal{E}}}
\newcommand{\cF}{{\mathcal{F}}}

\newcommand{\cN}{{\mathcal{N}}}

\newcommand{\tcS}{{\tilde{\cS}}}

\newcommand{\tpsi}{{\tilde{\psi}}}

\newcommand{\cM}{{\mathcal{M}}}
\newcommand{\cX}{{\mathcal{X}}}
\newcommand{\cY}{{\mathcal{Y}}}

\newcommand{\cS}{{\mathcal{S}}}

\newcommand{\bx}{{\boldsymbol{x}}}

\newcommand{\bz}{{\boldsymbol{z}}}

\newcommand{\bg}{{\boldsymbol{g}}}

\newcommand{\bth}{{\boldsymbol{\theta}}}

\newcommand{\bP}{{\boldsymbol{P}}}
\newcommand{\bC}{{\boldsymbol{C}}}

\newcommand{\ba}{{\boldsymbol{a}}}
\newcommand{\bb}{{\boldsymbol{b}}}

\newcommand{\be}{{\boldsymbol{e}}}

\newcommand{\bbR}{{\mathbb{R}}}

\newcommand{\bbI}{{\mathbb{I}}}

\DeclareMathOperator{\tr}{\textup{Tr}}

\usepackage{hyperref}
\hypersetup{
    colorlinks = true,
    linkcolor = {blue},
    citecolor = {blue},
    urlcolor = {blue},
}
\usepackage[utf8]{inputenc} % This defines the font-encoding you prefer to userep

\begin{document}
\title{Learning Quantum Data Distribution via Chaotic Quantum Diffusion Model}

\author{Quoc Hoan Tran}
\email{tran.quochoan@fujitsu.com}
\affiliation{Quantum Laboratory, Fujitsu Research, Fujitsu Limited, Kawasaki, Kanagawa 211-8588, Japan}

\author{Koki Chinzei}
\affiliation{Quantum Laboratory, Fujitsu Research, Fujitsu Limited, Kawasaki, Kanagawa 211-8588, Japan}

\author{Yasuhiro Endo}
\affiliation{Quantum Laboratory, Fujitsu Research, Fujitsu Limited, Kawasaki, Kanagawa 211-8588, Japan}

\author{Hirotaka Oshima}
\affiliation{Quantum Laboratory, Fujitsu Research, Fujitsu Limited, Kawasaki, Kanagawa 211-8588, Japan}

\date{\today}

\begin{abstract}
Generative models for quantum data pose significant challenges but hold immense potential in fields such as chemoinformatics and quantum physics. Quantum denoising diffusion probabilistic models (QuDDPMs) enable efficient learning of quantum data distributions by progressively scrambling and denoising quantum states. However, existing implementations typically rely on circuit-based random unitary dynamics, which can be costly to implement and sensitive to control imperfections, particularly on analog quantum hardware. We propose the chaotic quantum diffusion model, a framework that generates projected ensembles via chaotic Hamiltonian time evolution, providing a flexible and hardware-compatible diffusion mechanism. Requiring only global, time-independent control, our approach substantially reduces implementation overhead across diverse analog quantum platforms while achieving accuracy comparable to QuDDPMs. This method improves trainability and robustness, broadening the applicability of quantum generative modeling.
\end{abstract}

\pacs{Valid PACS appear here}

\maketitle

\section{Introduction}

Generative models aim to synthesize diverse data by learning their underlying probability distributions and play a central role in data understanding, simulation, and discovery. Quantum circuits are capable of generating classically intractable probability distributions, motivating the use of quantum systems as generative models that may outperform their classical counterparts in certain regimes~\cite{Boixo:2018:natphys,Arute:2019:nature}. Various quantum generative models have been proposed for classical data generation, including quantum generative adversarial networks (QuGANs)~\cite{lloyd:qugan:2018,zoufal:qugan:2019,huang:qugan:2021}, with style-based variants that improve image generation and drug design~\cite{baglio:2026:latent:qWass,bravo:2022:style}. Others include quantum variational autoencoders (QVAEs)~\cite{khoshaman:qvae:2018,wu:qvae-mole:2024}, tensor-network-based models~\cite{wall:tensor:2021}, and diffusion-based quantum models~\cite{cacioppo:diffusion:2023,parigi:diffusion:2024}. However, a clear empirical advantage over classical generative models remains unestablished. Recent results nonetheless point to possible regimes where a quantum advantage may emerge, such as an empirical capacity advantage in the number of trainable parameters~\cite{liepelt:2026:capacity} or a proposed exponential memory advantage in processing massive classical data~\cite{zhao:2026:advantage}.

More recently, growing attention has shifted toward quantum machine learning (QML) for quantum data, where the data themselves originate from quantum systems~\cite{editorial:2023:qmladv:nature}. Unlike classical data such as text or images, quantum data inherently encode quantum correlations, entanglement, and measurement statistics, posing fundamental challenges that classical generative models alone cannot address. Efficient quantum data generation is essential for advancing our understanding of quantum many-body systems and for applications in chemistry, biology, and materials science. Since quantum data generation intrinsically relies on quantum resources, quantum-native generative models are a natural and necessary choice.

Existing quantum generative models such as QuGANs and QVAEs can be used to prepare individual target quantum states~\cite{niu:equgan:2022,kim:hqugan:2024,tran:qvae:2024}, but they are generally inefficient at learning and generating ensembles of quantum states~\cite{beer:dqugan:2021}. This limitation arises mainly from the need to train deep variational quantum circuits (VQCs), which face optimization challenges such as barren plateaus~\cite{mcclean:barren:2018}. 
For broader context on assumptions, evaluation pitfalls, and limitations in contemporary QML, including quantum generative modeling, we refer to the recent review in Ref.~\cite{chang:2025:primerQML}.

Diffusion-based approaches provide a particularly promising route to ensemble learning. Early quantum extensions of the classical diffusion paradigm to quantum data have been formulated from complementary perspectives. A quantum-noise-driven paradigm has been introduced in which the forward diffusion is implemented by quantum channels (quantum Markov chains) in discrete time, or by stochastic Schrödinger-equation dynamics in continuous time~\cite {parigi:diffusion:2024}. This paradigm systematizes three variants depending on whether the forward and backward processes are realized classically or quantum mechanically.
The quantum denoising diffusion probabilistic model (QuDDPM)~\cite{zhang:qddpm:prl:2024} offers a promising approach to addressing issues in quantum generative models. QuDDPM learns a sequence of intermediate distributions that interpolate smoothly between a structured target distribution $\mathcal{E}_0$ and a maximally scrambled distribution, using quantum scrambling for the forward diffusion process and measurement-enabled VQCs for backward denoising. This layerwise training strategy improves trainability and mitigates barren plateaus.
However, existing QuDDPM implementations rely on circuit-based scrambling using random unitary circuits (RUCs), which require fine-grained spatio-temporal control and substantial circuit depth, with high costs in calibration and compilation.
These requirements pose significant implementation challenges, particularly on analog quantum platforms where control is typically global and time-independent. 

In this work, we propose a chaotic quantum diffusion model that replaces circuit-based scrambling with chaotic Hamiltonian time evolution. By generating projected ensembles through quantum chaos, our approach provides a flexible, hardware-compatible diffusion mechanism that eliminates the need for RUCs and enables efficient implementation across a broad range of quantum hardware.

\begin{figure*}
\centerline{\includegraphics[width = 1.0\linewidth]{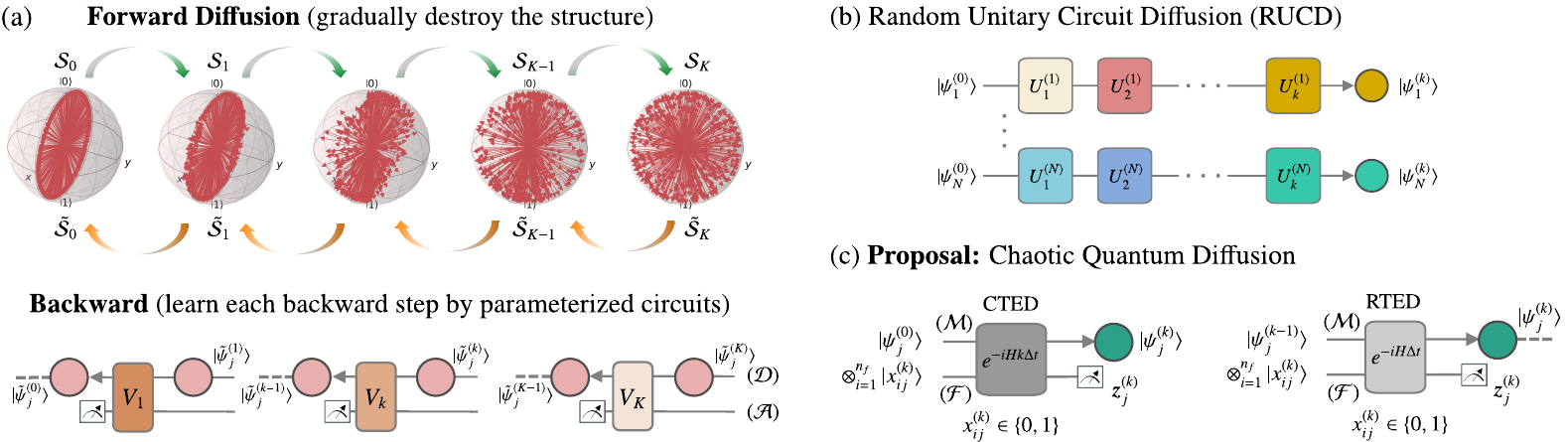}}
\caption{(a) General scheme of the quantum denoising diffusion probabilistic model (QuDDPM). 
The forward process applies quantum operations to initial quantum states $| \psi^{(0)}_j \rangle$ in the training dataset $S_0$, generating progressively noisier ensembles $S_k$ over $K$ steps. The backward process starts from random states $\{| \tilde{\psi}^{(K)}_j \rangle\}$ and trains parameterized quantum circuits $V_k$ on the data system $\mathcal{D}$ ($n_d$ qubits) coupled with an ancilla system $\mathcal{A}$ ($n_a$ qubits), followed by measurements on $\mathcal{A}$ to iteratively reconstruct the target distribution.
(b) Random unitary circuit diffusion (RUCD), where sequences of random unitary gates are applied to each training instance. 
(c) Chaotic quantum diffusion (our proposal) via the projected ensemble framework. Here, a chaotic and fixed Hamiltonian $H$ governs time evolution on a bipartite system comprising the data subsystem $\cM$ ($n_d$ qubits) and the complement subsystem  $\cF$ (initialized by random computational basis states with $n_f$ qubits).
Projective measurements on $\cF$ yield a projected ensemble on $\cM$.
In creating the diffused state ensemble $\cS_k$, applying cumulative evolution $k\Delta t$ to $\cS_0$ defines cumulative time evolution diffusion (CTED), 
while applying repeated time evolution $\Delta t$ to $\cS_{k-1}$ defines repeated time evolution diffusion (RTED).
}
\label{fig:overview}
\end{figure*}

\section{Quantum Denoising Diffusion Probabilistic Model (QuDDPM)}
%Generative models are powerful tools capable of producing data from a target domain and evaluating the likelihood of given data points. 
We address the problem of learning quantum data distributions and explain the idea of the QuDDPM model in the literature [Fig.~\ref{fig:overview}(a)].
Consider a training dataset $\cS$ of $N$ independent quantum states drawn from an unknown probability distribution $\cE_0$. A generative model is characterized by a parameterized probability distribution $\cE_\bth$, from which samples can be drawn. If $\cE_\bth$ is realized through VQCs, $\bth$ corresponds to the adjustable parameters of these circuits.
The training goal is to obtain a distribution $\cE_\bth$ that closely approximates $\cE_0$. We optimize the distance $D(\cE_\bth, \cE_0)$ such as the Maximum Mean Discrepancy (MMD) and Wasserstein distance between these distributions. Since $D(\cE_\bth, \cE_0)$ cannot be computed directly, we take a dataset $\cS_\bth$ sampling from $\cE_\bth$ and compute the distance $D(\cS, \cS_\bth)$. In the inference phase, we fix the trained $\bth$ and generate new states $\ket{\psi} \sim \cE_\bth$. 

Unlike prior quantum models that focus on generating average states, such as fully mixed states, QuDDPM targets entire state distributions, enabling the generation of correlated or structured quantum data. Building on the success of classical DDPMs in tasks such as image generation, QuDDPM extends this framework to quantum data.
Classical DDPMs add Gaussian noise in a forward process to gradually corrupt data toward a noise distribution (e.g., an isotropic Gaussian), then use learned denoising in the reverse process. QuDDPM replaces Gaussian noise with scrambling RUCs in the forward process, applying random unitaries that evolve quantum states toward a near–Haar-random distribution. This choice is more natural for quantum systems, as Gaussian noise does not directly correspond to physically valid unitary dynamics.

In the backward process, QuDDPM uses measurement-enabled denoising with VQCs and ancilla qubits via projective measurements. The ancilla qubits are entangled with the system qubits via controlled, parameterized operations, and subsequent measurements collapse the ancilla states, effectively denoising by projecting the system onto subspaces that approximate the target distribution.
These measurement-induced, non-unitary operations are essential for breaking the constraints of purely unitary quantum circuits and enabling flexible generative modeling. This contrasts with fully unitary quantum models, which are less expressive when learning inherently non-unitary processes such as noise.

\subsection{Forward diffusion process}

In the forward diffusion process using RUCs, $K$ random unitary gates $U_1^{(j)},\ldots, U_K^{(j)}$ are applied to each sample $\ket{\psi^{(0)}_j}$ ($j=1,\ldots,N$) in the training dataset $\cS$ [Fig.~\ref{fig:overview}(b)]. This evolves the ensemble $\cS_k =\left\{ \ket{\psi^{(k)}_j} = \prod_{l=1}^k U_l^{(j)}\ket{\psi^{(0)}_j} \right\}_j$ toward a Haar-random states ensemble over the Hilbert space. We call this scheme RUC diffusion (RUCD). Assuming the execution time for each $U^{(j)}_k$ is $\tau_u$, the execution time to generate each $\cS_k$ is $\tau_u Nk$. Therefore, the RUCD requires $NK$ random unitary gates with $\tau_u N\sum_{k=1}^K k= \tau_u NK(K+1)/2$ execution time to generate all $\cS_k$.

The forward diffusion process does not require exact Haar randomness. Recent studies have shown that shallow random circuits can efficiently approximate unitary $m$-designs with depth scaling polynomially or even logarithmically in system size for small $m$~\cite{schuster:RUC:science}. Such approximate designs are sufficient to rapidly scramble local observables and correlations, making them suitable for diffusion-based generative modeling. QuDDPM's implementation relies on approximate scrambling rather than high-fidelity Haar sampling.
Nevertheless, even shallow circuit constructions for approximate designs require explicit gate-level control and circuit compilation, which can be costly or infeasible on analog quantum platforms. This motivates alternative diffusion mechanisms that naturally generate effective scrambling without the need for circuit-based random unitaries.

\subsection{Backward denoising process}

The backward process starts with an ensemble $\tcS_K = \{\ket{\tpsi^{(K)}_j}\}_j$ sampled from Haar-random states and reduces noise step by step.
In practice, since it is difficult to sample from a complete Haar-random state distribution, $\tcS_K = \{\ket{\tpsi^{(K)}_j}\}_j$ are sampled from the product of single-qubit Haar-random states.
As depicted in Fig.~\ref{fig:overview}(a), the denoising step applies a parameterized unitary $V_k = V(\bth_k)$ to the data system $\cD$ (input $\ket{\tpsi^{(k)}_j}$) and $n_a$ ancilla qubits in the ancilla system $\cA$ ($\ket{\boldsymbol{0}}_\cA$), followed by projective measurements in the computational basis on $\cA$, yielding the state $\ket{\tpsi^{(k-1)}_j}$ in $\cD$. 
Assuming the measurement outcome $\bz^{(k)}_j$ is obtained on the ancilla, this operation is formulated as 
\begin{align}
\Phi^{(k)}_j(\ket{\tpsi^{(k)}_j}) &= \dfrac{(\bbI_\cD \otimes \Pi_\cA)V_k \ket{\tilde{\Psi}^{(k)}_j} }{\sqrt{\bra{\tilde{\Psi}^{(k)}_j}V^\dagger_k(\bbI_\cD \otimes \Pi_\cA) V_k\ket{\tilde{\Psi}^{(k)}_j}}}\\
&= \ket{\tpsi^{(k-1)}_j} \otimes \ket{\bz^{(k)}_j}_\cA,
\end{align}
where $\Pi_\cA = {\ket{\bz^{(k)}_j}\bra{\bz^{(k)}_j}}_\cA$ and $\ket{\tilde{\Psi}^{(k)}_j} = \ket{\tpsi^{(k)}_j} \otimes \ket{\boldsymbol{0}}_\cA$.

Training involves $K$ cycles with a layerwise scheme. At the cycle $(K-k+1)$ (with $k=K,\ldots,1$), the forward diffusion with $U^{(j)}_1$ to $U^{(j)}_{k-1}$ generates the noisy ensemble $\cS_{k-1} = \left\{ \ket{\psi^{(k-1)}_j}  \right\}_j$.
Parameters of $V(\bth_k)$ are optimized to make the denoised ensemble $\tcS_{k-1} = \left\{\ket{\tpsi^{(k-1)}_j}\right\}_j$ approximate $\cS_{k-1}$ via the minimization of the cost function $D(\cS_{k-1}, \tcS_{k-1})$.
After optimization, the parameters $\bth_k$ are fixed for use in the next cycle to optimize $\bth_{k-1}$.
This layerwise training approach divides the original training problem into $K$ manageable sub-tasks, ensuring convergence for incremental distribution transitions and mitigating issues such as barren plateaus~\cite{mcclean:barren:2018}.

%\subsection{Cost Function}
The cost function in QuDDPM measures the similarity between two quantum state ensembles using
a symmetric, positive definite quadratic kernel $\kappa(\ket{\mu}, \ket{\phi})$. This kernel can be defined by the state fidelity computed via the SWAP test (Fig.~\ref{fig:SWAP}) or directly derived from the classical shadows kernel~\cite{Huang:2020:shadows, chinzei:2025:shadows}.
%, or the classical shadows kernel, which enables efficient classical-based computation using classical shadows data.
We consider two cost functions $D_{\textrm{MMD}}$ and $D_{\textrm{Wass}}$, corresponding to the MMD distance and the 1-Wasserstein distance based on the state fidelity.

The MMD distance between two state ensembles $\cX=\{\ket{\mu_i}\}$ and $\cY=\{\ket{\psi_j}\}$ is defined as
\begin{align}
    \cD_{\textrm{MMD}}(\cX, \cY) = \bar{\kappa}(\cX, \cX) + \bar{\kappa}(\cY, \cY) - 2\bar{\kappa}(\cX, \cY),
\end{align}
where $\bar{\kappa}(\cX, \cY) = \mathbb{E}_{\ket{\mu} \in \cX, \ket{\phi} \in \cY}[\kappa(\ket{\mu}, \ket{\phi})]$.

The 1-Wasserstein distance is presented as an enhancement when the MMD distance is not feasible for distinguishing two state ensembles. Given normalized $\kappa$ (i.e., $\kappa(\ket{\phi}, \ket{\phi}) = 1 \forall \ket{\phi}$), the pairwise cost matrix $\bC = (C_{i,j}) \in \bbR^{|\cX|\times |\cY|}$ is computed as $C_{i,j} = 1 - \kappa(\ket{\mu_i}, \ket{\psi_j})$. The 1-Wasserstein distance is calculated via the formulation of the optimal transport problem as a linear programming procedure to find the optimal transport plan $\bP = (P_{i,j}) \in \bbR^{|\cX|\times |\cY|}$:
\begin{align}
    \cD_{\textrm{Wass}}(\cX, \cY) = \min_{\bP}\sum_{i,j}P_{i,j}C_{i,j},\\
    \textrm{s.t. } \bP\boldsymbol{1}_{|\cY|} = \ba, \quad \bP^{\top}\boldsymbol{1}_{|\cX|} = \bb, \quad \bP \geq 0. 
\end{align}
Here, $\boldsymbol{1}_{|\cX|}$ and $\boldsymbol{1}_{|\cY|}$ are all-ones vectors with size $|\cX|$ and $|\cY|$, respectively, and $\ba \in \bbR^{|\cX|}$ and $\bb \in \bbR^{|\cY|}$ are the probability vectors histogram corresponding to $\cX$ and $\cY$. Normally, we set uniform histograms as $\ba = \tfrac{1}{|\cX|} \boldsymbol{1}_{|\cX|}$ and $\bb = \tfrac{1}{|\cY|} \boldsymbol{1}_{|\cY|}$.

% The MMD distance between two state ensembles $\cX=\{\ket{\mu_i}\}$ and $\cY=\{\ket{\psi_j}\}$ is defined as
% $
%     D_{\textrm{MMD}}(\cX, \cY) = \bar{\kappa}(\cX, \cX) + \bar{\kappa}(\cY, \cY) - 2\bar{\kappa}(\cX, \cY),
% $
% where $\bar{\kappa}(\cX, \cY) = \mathbb{E}_{\ket{\mu} \in \cX, \ket{\phi} \in \cY}[\kappa(\ket{\mu}, \ket{\phi})]$.
% The 1-Wasserstein distance $D_{\textrm{Wass}}(\cX, \cY)$ is further presented as an enhancement in the situation where the MMD distance is not feasible to distinguish two state ensembles. Given normalized $\kappa$ (i.e., $\kappa(\ket{\phi}, \ket{\phi}) = 1 \forall \ket{\phi}$), the pairwise cost matrix $\bC = (C_{i,j}) \in \bbR^{|\cX|\times |\cY|}$ is computed as $C_{i,j} = 1 - \kappa(\ket{\mu_i}, \ket{\psi_j})$. The 1-Wasserstein distance is calculated via the formulation of the optimal transport problem into a linear programming procedure to find the optimal transport plan.

\section{Proposal: Chaotic Quantum Diffusion}

RUCD requires significant computational overhead to design random gate sequences, making it unsuitable for many quantum systems, particularly analog platforms such as Rydberg atom arrays and ultracold atoms in optical lattices. 
These platforms excel at simulating continuous-time dynamics using global Hamiltonians but lack the fine-grained, time-dependent control required for arbitrary gate sequences. Implementing RUCD on analog hardware would require digitizing the process via Trotterization or pulse shaping, introducing approximation errors and increasing susceptibility to noise from environmental couplings.

We address this challenge by adopting the projected ensemble framework~\cite{choi:nature:2023,cotler:prxquantum:2023,mok:optimal:2025}, which utilizes a single chaotic many-body wave function to generate a random ensemble of pure states on a subsystem. 
In this framework, projective measurements are performed on the larger subsystem of a bipartite state undergoing quantum chaotic evolution. This process yields a set of pure states on the smaller subsystem, accompanied by their respective Born probabilities. Collectively, these states form the projected ensemble, which converges to a state design when the measured subsystem is sufficiently large. In this context, a chaotic Hamiltonian is characterized by non-integrability, ergodic dynamics, spectral properties resembling those of random matrix theory, and strong entanglement generation.
Recent works have further clarified the role of symmetry in the Hamiltonian and measurement structure in the emergence of state designs from projected ensembles~\cite{varikuti:2024:unravel}.

Our approach offers significant advantages over RUCD because it requires only global and time-independent control. It is more accessible and adaptable to a wide range of quantum systems, including analog platforms, where implementing precise gate sequences is challenging.
We emphasize that the requirement of only global, time-independent control is a decisive
advantage on analog platforms, where arbitrary, finely time-resolved gate sequences are
difficult to realize. On gate-based digital platforms this specific control argument is
weaker, since arbitrary gates are natively available. However, reusing a single fixed
Hamiltonian evolution still removes the per-step circuit-synthesis and recompilation overhead
incurred by random unitary circuits.

\subsection{Projected Ensemble}
We consider a many-body system partitioned into a subsystem $\cM$ (with $n_d$ qubits) and its complement $\cF$ (with $n_f$ qubits). Given a generator state $\ket{\Phi}$ on the total system $\cM+\cF$, which is produced by chaotic evolution,  we perform local measurements on $\cF$, typically in the computational basis. This yields different pure states $\ket{\Phi_\cM(\bz_\cF)}$ on $\cM$, each corresponding to a distinct measurement outcome $\bz_\cF$, which are bitstrings of the form, for example, $\bz_\cF=001\ldots010$. 
The collection of these states, together with probabilities $p(\bz_\cF) = \| \left( \bbI_\cM \otimes \bra{\bz_\cF}\right) \ket{\Phi}\|^2$, forms the projected ensemble on $\cM$:  $\left\{ \ket{\Phi_\cM(\bz_\cF)}, p(\bz_\cF)\right\}_{\bz_{\cF}}$.
The projected ensemble provides a full description of the total system state as
$\ket{\boldsymbol{\Phi}} = \sum_{\bz_\cF} \sqrt{p(\bz_\cF)} \ket{\Phi_\cM(\bz_\cF)} \otimes \ket{\bz_\cF}$.
The ensemble of pure states is not just a density matrix. The density matrix represents the average mixed state over the distribution and captures first-order statistics such as expectation values of observables. Therefore, different ensembles can lead to the same density matrix. These ensembles can be distinguished through their higher moments of the observable expectation values.

Chaotic dynamics scramble quantum information, producing universal correlations and randomness within the subsystem. For a case of infinite-temperature thermalization, with sufficiently large $n_f = \Omega(m n_d)$ and the generator state $\ket{\Phi}$ obtained by quenched time evolution of chaotic Hamiltonians, the projected ensemble approximates $m$-design of Haar-random states~\cite{choi:nature:2023,cotler:prxquantum:2023}.
This convergence is quantified by the trace distance $\frac{1}{2} \| \rho^{(m)}_E - \rho^{(m)}_{\textrm Haar}\|_1 \to 0$ for the $m$-th moment operators $\rho^{(m)}_E$ and $\rho^{(m)}_{\textrm Haar}$ calculated from the projected ensemble and the Haar-random states ensemble, respectively.
This relies on the eigenstate thermalization hypothesis and the concentration of measure: the chaotic evolution scrambles quantum information universally, making higher moments converge to the Haar measure.

Inspired by this framework, we propose two schemes in which the diffusion is implemented through chaotic Hamiltonian evolution [Fig.~\ref{fig:overview}(c)].
%A comparison of resources between RUCD and our proposed methods is summarized in Table~\ref{tab1}.
%We further introduce the classically-enhanced projected ensemble to accelerate the diffusion process.

\subsection{Cumulative Time Evolution Diffusion (CTED)}
For each $\ket{\psi^{(0)}_j} \in \cS_0$ on $\cM$, we implement a diffusion process to obtain the $k$-step state $\ket{\psi^{(k)}_j} \in \cS_k$.
The initial state $\ket{\bx^{(k)}_j}$ in $\cF$ is randomly sampled from the ensemble $\{\ket{\bx}, q(\bx)\}_{\bx \in \{0,1\}^{n_f}}$ of computational basis states with probability $q(\bx)$. 
Inspired by Ref.~\cite{mok:optimal:2025}, we introduce $q(x)$ here to increase classical randomness injected during the protocol, which is converted into quantum randomness of the resulting state ensemble.
We use the uniform distribution in our experiments because it is the maximum-entropy choice over the computational basis and therefore injects the most classical randomness per qubit, making it the most efficient option for driving the ensemble toward its universal form. 
A more concentrated $q(\bx)$ injects less randomness and can bias the projected ensemble toward its own structure when scrambling is incomplete. 
We regard $q(\bx)$ as a tunable knob rather than a fixed requirement, since in the strongly-chaotic limit the projected ensemble approaches the same universal distribution largely independently of $q(\bx)$.

The input $\ket{\psi^{(0)}_j}\otimes \ket{\bx^{(k)}_j}$ of the system $(\cM+\cF)$ is evolved cumulatively  under the same unitary $e^{-iHk\Delta t}$ for all data indexes $j$, where $H$ is a fixed chaotic Hamiltonian~\cite{cotler:prxquantum:2023}.
Subsequently, a random local measurement is performed on $\cF$, yielding the measurement record $\bz^{(k)}_j$ with the probability $p_\bx(\bz^{(k)}_j)=q(\bx_j^{(k)})\left\| \left( \bbI_\cM \otimes \bra{\bz^{(k)}_j}\right) e^{-iHk\Delta t}\left(\ket{\psi_j^{(0)}} \otimes \ket{\bx_j^{(k)}}\right)\right\|^2$, and the resulting state $\ket{\psi^{(k)}_j}$ on $\cM$.
In this case, we obtain the classically-enhanced projected ensemble $\cE^{(k)} = \{q(\bx_j^{(k)})p_\bx(\bz_j^{(k)}); \ket{\psi_\bx(\bz_j^{(k)})}\}_j$.

Since the execution time for $e^{-iHk\Delta t}$ generally scales with $k$, we write it as $k\tau_c$, where $\tau_c$ denotes the execution time per unit of evolution time. This CTED requires $K$ unitaries and $\tau_c N \sum_{k=1}^{K} k = \tau_c NK(K+1)/2$ execution time to generate all $\cS_k$.

\begin{figure}
  %\vspace{-\intextsep} % pull up to align nicely with the text (optional)
  \centering
  \includegraphics[width=\linewidth]{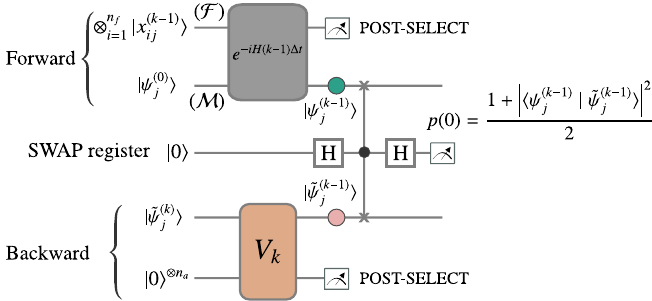}
  \caption{The schematic circuit to compute the fidelity between the forward state $\ket{\psi_j^{(k-1)}}$ and the denoised state $\ket{\tilde{\psi}_j^{(k-1)}}$ using the SWAP test.}
  \label{fig:SWAP}
\end{figure}

In Fig.~\ref{fig:SWAP}, we depict the schematic circuit to compute 
the state fidelity between the forward state $\ket{\psi_j^{(k-1)}}$ and the denoised state $\ket{\tilde{\psi}_j^{(k-1)}}$ can be computed using the SWAP test.
The SWAP test consists of two Hadamard gates and a controlled-swap gate applied on $2n_d+1$ qubits. The probability of measure 0 on the SWAP register is $p(0)=\frac{1}{2} + \frac{1}{2} \left| \braket{\psi_j^{(k-1)} \mid  \tilde{\psi}_j^{(k-1)}} \right|^2$, which directly derives the fidelity.
We apply a post-selection protocol to the projective measurements in the forward and backward processes, ensuring $\ket{\psi_j^{(k-1)}}$ and $\ket{\tilde{\psi}_j^{(k-1)}}$ remain consistent across SWAP test measurements.

\subsection{Repeated Time Evolution Diffusion (RTED)}
Unlike the CTED, we consider the input of $(\cM + \cF)$ as $\ket{\psi^{(k-1)}_j}\otimes \ket{\bx^{(k)}_j}$, where $\ket{\psi^{(k-1)}_j} \in \cS_{k-1}$, then evolve this state under the same unitary $e^{-iH\Delta t}$ for all data indexes $j$.
We then repeat this process for $k=1,\ldots, K$.
Assuming the execution time for $e^{-iH\Delta t}$ is $\tau_r$, to generate each $\cS_k$, we need to repeat $e^{-iH\Delta t}$ for $k$ times, making the execution time is $k\tau_r$.
Therefore, this RTED requires only one unitary but $\tau_r N \sum_{k=1}^K k = \tau_r NK(K+1)/2$ execution time to generate all $\cS_k$.

In contrast to CTED, which requires a distinct cumulative propagator $e^{-iHk\Delta t}$ for each step $k$, RTED reuses a single short-time unitary $e^{-iH\Delta t}$. On a digital device, the CTED propagator $e^{-iHk\Delta t}$ can be realized as the first-order Trotter product $(e^{-iH\Delta t})^k$. This Trotterization applies to the coherent evolution within a single diffusion step. However, the diffusion process itself is not a single Trotter trajectory, since each diffused ensemble $\cS_k$ is prepared independently from $\cS_0$ and the projective measurement on $\cF$ intervenes between steps.

\subsection{Universal Ensemble}
In the CTED and RTED schemes, the generator state at each diffusion step is generally not at infinite temperature. As a result, the forward diffusion operates in a finite-temperature regime where convergence to the Haar measure is not expected~\cite{cotler:prxquantum:2023}. In CTED, longer evolution times ($k\Delta t$) allow more scrambling.
However, because the dynamics start from finite-energy states, the resulting projected ensembles remain constrained by energy conservation and saturate at a finite-temperature universal distribution. In RTED, resets on diffused states preserve some initial structure across steps, further deviating from full Haar scrambling.
If the target distribution $\cE_0$ were already Haar-random, the projected ensembles would converge to approximate $m$-designs, but in the general case, the target distributions are often low-entropy, making the projected ensemble finite-temperature-like.

Recent work has shown that, under chaotic many-body dynamics at finite energy density, projected ensembles generically converge to Scrooge ensembles~\cite{mcginley:2025:scrooge:ensemble,scrooge:2025:prb}. These ensembles are universal in that their moments for local observables are fixed by the Eigenstate Thermalization Hypothesis (ETH)~\cite{scrednicki:ETH:1994, alessio:ETH:2016, Deutsch:2018:ETH} and are therefore independent of the microscopic details of the chaotic Hamiltonian, while remaining maximally constrained by conserved quantities such as energy~\cite{mok:2026:scrooge}. From an information-theoretic viewpoint, a Scrooge ensemble is the maximum-entropy distribution consistent with these macroscopic constraints. It embodies the minimal randomness the dynamics must generate, rather than the full Haar randomness of an unconstrained system.

The Scrooge-like nature of the diffuse ensemble affects backward denoising but does not hinder trainability. In contrast to RUCD, where diffusion toward fully Haar-random states can lead to barren-plateau behavior in deep variational quantum circuits, the energy-dependent structure retained in finite-temperature diffusion may provide a smoother learning landscape. Our numerical results demonstrate that, despite deviating from Haar scrambling, CTED and RTED achieve generative accuracy comparable to RUCD while operating in a more physically realistic and information-efficient diffusion regime.

% \begin{table*}[htbp]
% \small
% \caption{Resource comparison between quantum diffusion methods with $K$ diffusion steps on $N$ samples of $n_d$-qubit data}
% \begin{center}
% \begin{tabular}{|c|c|c|c|}
% \hline
%  & \textbf{{(Prev.) RUCD}}& \textbf{{(Our) CTED}}& \textbf{{(Our) RTED}} \\
% \hline
% Number of unitary types 
% & $NK$ & $K$ & $1$ \\
% \hline
% Number of ancilla qubits
% & $0$ & $\Omega(n_d)$ & $\Omega(n_d)$ \\
% \hline
% \begin{tabular}{c}
% \\Execution time of the forward diffusion \\ ~
% \end{tabular}
%  & $\dfrac{\tau_u NK(K+1)}{2}$ & $\dfrac{\tau_c NK(K+1)}{2}$ & $\dfrac{\tau_r NK(K+1)}{2}$ \\
% \hline
% \end{tabular}
% \label{tab1}
% \end{center}
% \end{table*}

\section{Results}

We conduct numerical experiments using three illustrative datasets: synthetic distributions of clustered quantum states, circular quantum states, and a quantum distribution derived from a chemistry dataset. We simulate the quantum circuits with the TensorCircuit library~\citep{zhang:2023:tensorcircuit}, and rely on JAX~\citep{jax:2018:github} for automatic differentiation to support gradient-based optimization. The circuit parameters are initialized uniformly within $[-\pi, \pi]$, and optimization is performed using the Adam algorithm with a learning rate of $0.001$.
Source code to reproduce these experiments is available in the GitHub repository~\cite{tran:2026:github}.

For CTED and RTED, the dynamics are governed by a one-dimensional mixed-field Ising Hamiltonian with open boundary conditions, defined on the total system comprising $n_d + n_f$ sites as
\begin{align}\label{eqn:Ising:Ham}
        H = \sum_{j=1}^{n_d+n_f} \left(h^x\sigma^x_j + h^y\sigma^y_j\right) + J\sum_{j=1}^{n_d+n_f-1}\sigma^x_j\sigma^x_{j+1}.
\end{align}
Here, $\sigma^\mu_j$ (with $\mu = x, y, z$) denotes the Pauli operators at site $j$, $J$ represents the interaction strength, and $h^x$ and $h^y$ are the strengths of the longitudinal and transverse magnetic fields, respectively. 
This choice is a canonical minimal model of quantum chaos. As a well-characterized Hamiltonian, it lets us attribute the observed behavior to the diffusion mechanism rather than incidental model-specific features. Its nearest-neighbor Ising couplings and global transverse and longitudinal fields also map directly onto interactions available in Rydberg-atom arrays and trapped-ion simulators, making it experimentally relevant. However, any sufficiently chaotic, hardware-compatible Hamiltonian is expected to serve equally well as a diffusion engine.

In the presence of a non-zero longitudinal field ($h^x \neq 0$), the Hamiltonian exhibits ergodic behavior~\cite{cotler:prxquantum:2023}, with its eigenvalues and eigenvectors conforming to the predictions of the ETH.
To quantify the chaoticity of the Hamiltonian, we use the average adjacent-gap ratio $\braket{r}$. Let $\{E_n\}$ be the eigenvalues of $H$ within a fixed symmetry sector, sorted in increasing order, and let $s_n = E_{n+1} - E_n \geq 0$ be the consecutive level spacings. The gap ratio at level $n$ is $r_n=\frac{\min(s_n,s_{n-1})}{\max(s_n,s_{n-1})}\in[0,1]$, and $\braket{r}$ denotes its average over the spectrum (discarding the spectral edges).
We compute $\braket{r}$ over the $(h^x, h^y)$ plane (Fig.~\ref{fig:Ising:map}), where $\braket{r} \approx 0.53$ signals Wigner-Dyson (chaotic) statistics and $\braket{r}\approx 0.39$ signals Poisson (integrable) statistics.
In our experiments, we adopt the operating point $(h^x, h^y) = (0.8090, 0.9045)$, $J = 1.0$, which lies well inside the chaotic region (marked by the star symbol in Fig.~\ref{fig:Ising:map}).

We discretize the time evolution using a time step $\Delta t=0.02$ in units of $J^{-1}$ over $K$ diffusion steps. This choice balances two requirements: $\Delta t$ should be small enough that successive diffused ensembles evolve smoothly, so that layerwise backward training involves well-conditioned incremental sub-tasks, while also being large enough for the total evolution time $K\Delta t$ to reach the finite-temperature steady regime within a practical number of steps $K$. When the fields $(h^x,h^y)$ are varied, $\Delta t$ can be rescaled by the spectral width of the Hamiltonian, keeping the per-step scrambling approximately fixed across parameter choices.

\begin{figure}
\centering
\includegraphics[width=1.0\linewidth]{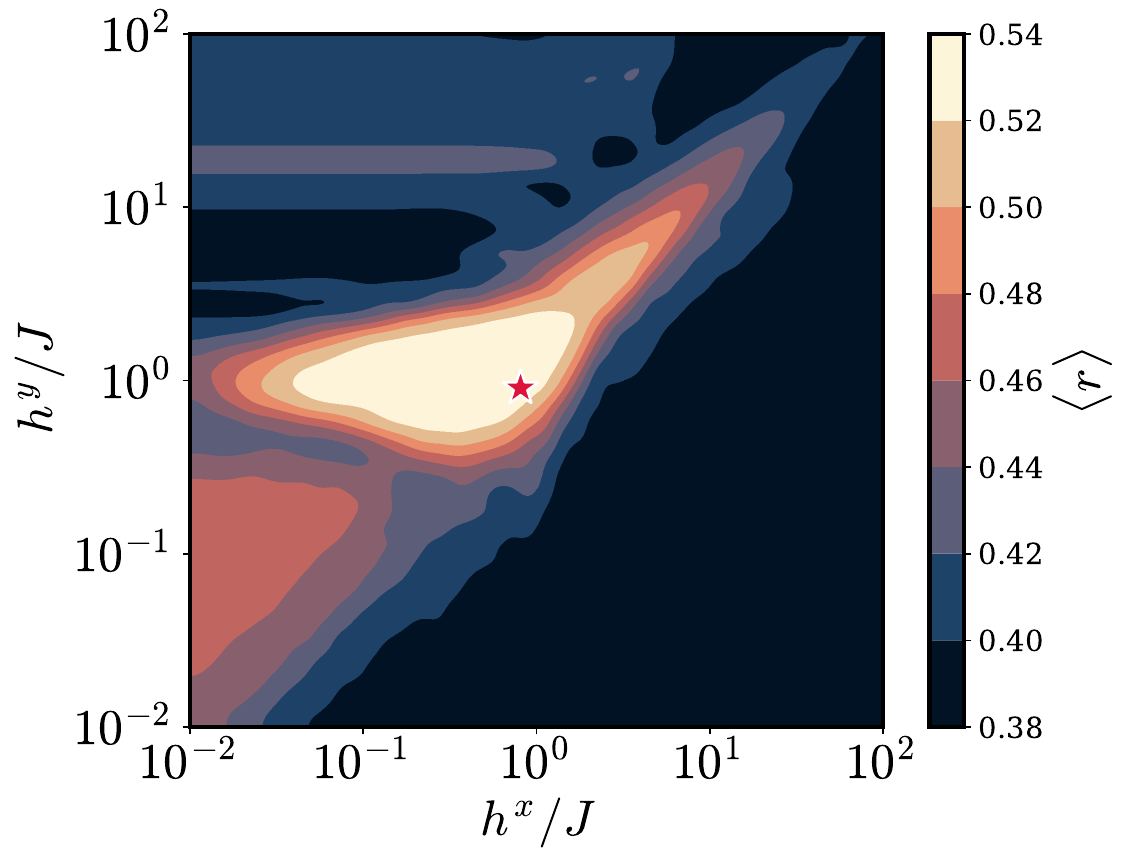}
\caption{Average adjacent-gap ratio $\braket{r}$ of the mixed-field Ising Hamiltonian in Eq.~\eqref{eqn:Ising:Ham} across the $(h^x, h^y)$ plane at $J=1$ and the total system size $n_d + n_f=14$. $\braket{r} \approx 0.53$ indicates chaotic and Wigner-Dyson statistics, where $\braket{r} \approx 0.39$ indicates integrable and Poisson statistics. The star marks the operating point $(h^x, h^y) = (0.8090, 0.9045)$, well inside the chaotic region.}
\label{fig:Ising:map}
\end{figure}

For RUCD, we adopt the fast scrambling circuits from Ref.~\cite{zhang:qddpm:prl:2024}, applying $\prod_{l=1}^k U^{(j)}_l$ to each initial state $\ket{\psi^{(0)}_j}$. Here, $U^{(j)}_l=\Omega_l(s^{(j)}_l) W_l(\bg^{(j)}_l)$, where $W_l(\bg^{(j)}_l)$ implements single-qubit rotations as
\begin{align}
    W_l(\bg^{(j)}_l) = \bigotimes_{q=1}^{n_d}e^{-ig^{(j)}_{l,3q-1}\frac{Z_q}{2}}e^{-ig^{(j)}_{l,3q-2}\frac{Y_q}{2}}e^{-ig^{(j)}_{l,3q-3}\frac{Z_q}{2}},
\end{align}
and an entangling layer $\Omega_l(s^{(j)}_l)$ applies ZZ rotations across all qubit pairs as
\begin{align}
    \Omega_l(s^{(j)}_l) = \prod_{q_1<q_2}\exp\left[\tfrac{-is^{(j)}_l}{2\sqrt{n_d}}Z_{q_1}Z_{q_2}\right].
\end{align}
The ranges of uniformly random angles $g^{(j)}_{l,q}$ and $s^{(j)}_l$ are tuned as $[-\alpha\pi/8,\alpha\pi/8]$ and $[0.4\alpha, 0.6\alpha]$, respectively. Here, $\alpha$ is the scaling parameter to control the diffusion speed, which is set as $\alpha = k^2/100$ at the $k$-th step.

In the backward process with $n_a$ ancilla qubits, each $V_k$ is constructed using a Hardware Efficient Ansatz on $n=n_d + n_a$ qubits with $L$ layers as
$V_k(\bth_k) = \prod_{l=1}^L \tilde{\Omega} \tilde{W}_l(\bth^{(l)}_k)$,
where
\begin{align}
    &\tilde{W}_l(\bth^{(l)}_k) = \bigotimes_{q=1}^{n}e^{-i\theta^{(l)}_{k,2q-1}\frac{Y_q}{2}} e^{-i\theta^{(l)}_{k,2q-2}\frac{X_q}{2}},\\
&\tilde{\Omega} = \bigotimes_{q=1}^{\lfloor (n-1)/2 \rfloor} CZ_{2q,2q+1} \bigotimes_{q=1}^{\lfloor n/2 \rfloor} CZ_{2q-1,2q}.
\end{align}

\subsection{Learning Multi-cluster and Circular Distribution}\label{subsec:cluster:circular}

\begin{figure*}
\centering
\includegraphics[width=1.0\linewidth]{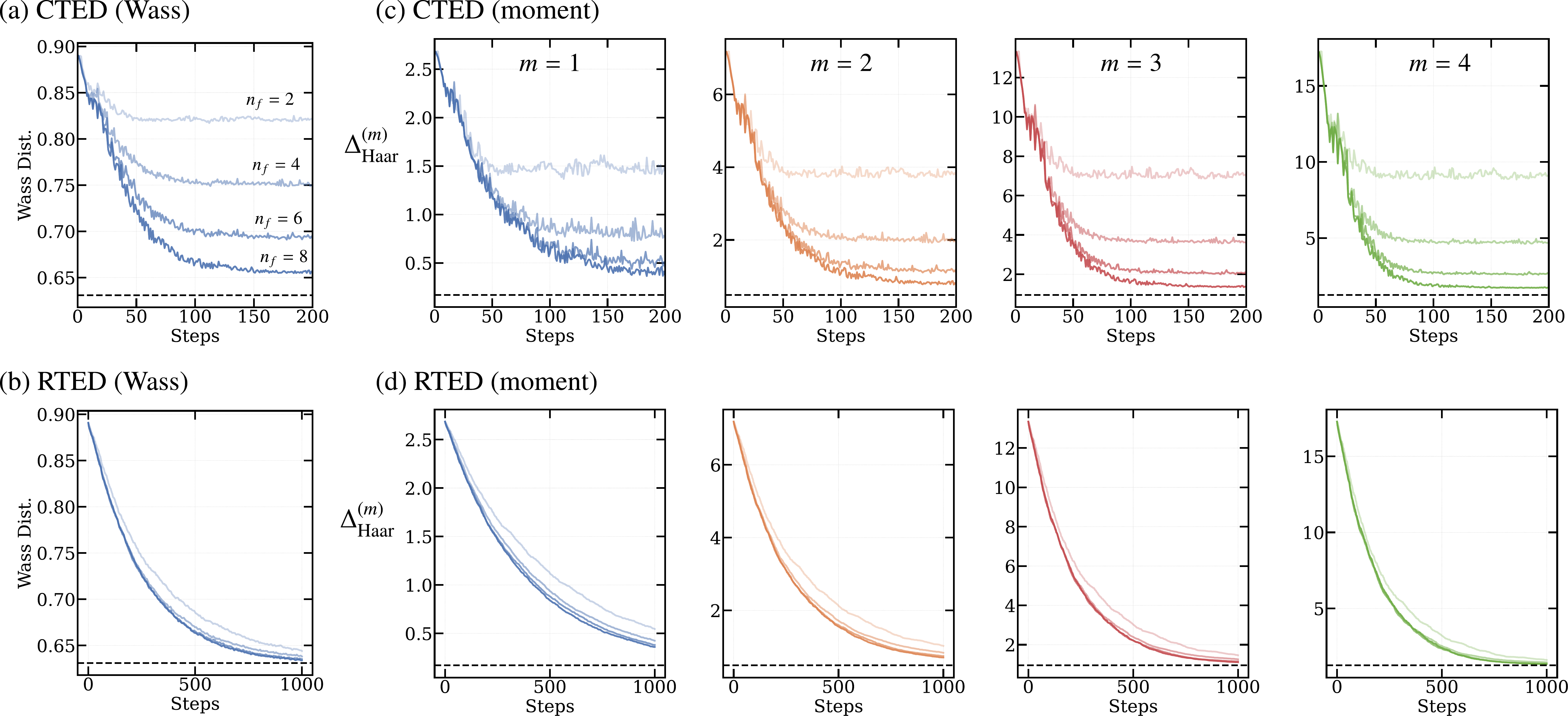}
\caption{1-Wasserstein distance and normalized moment distance $\Delta_{\textrm{Haar}}^{(m)}(k)$ between the diffused ensemble and the Haar-random ensemble versus diffusion step $k$, for CTED and RTED. Results are shown for different complement sizes $n_f \in \{2, 4, 6, 8\}$ and moment orders $m$. The dashed horizontal line in each panel is the finite-sample Haar floor. It depends only on the sample size $N$ and the moment order $m$ and is independent of $n_f$.}
\label{fig:evol}
\end{figure*}

\begin{figure*}
\centering
    \includegraphics[width=1.0\linewidth]{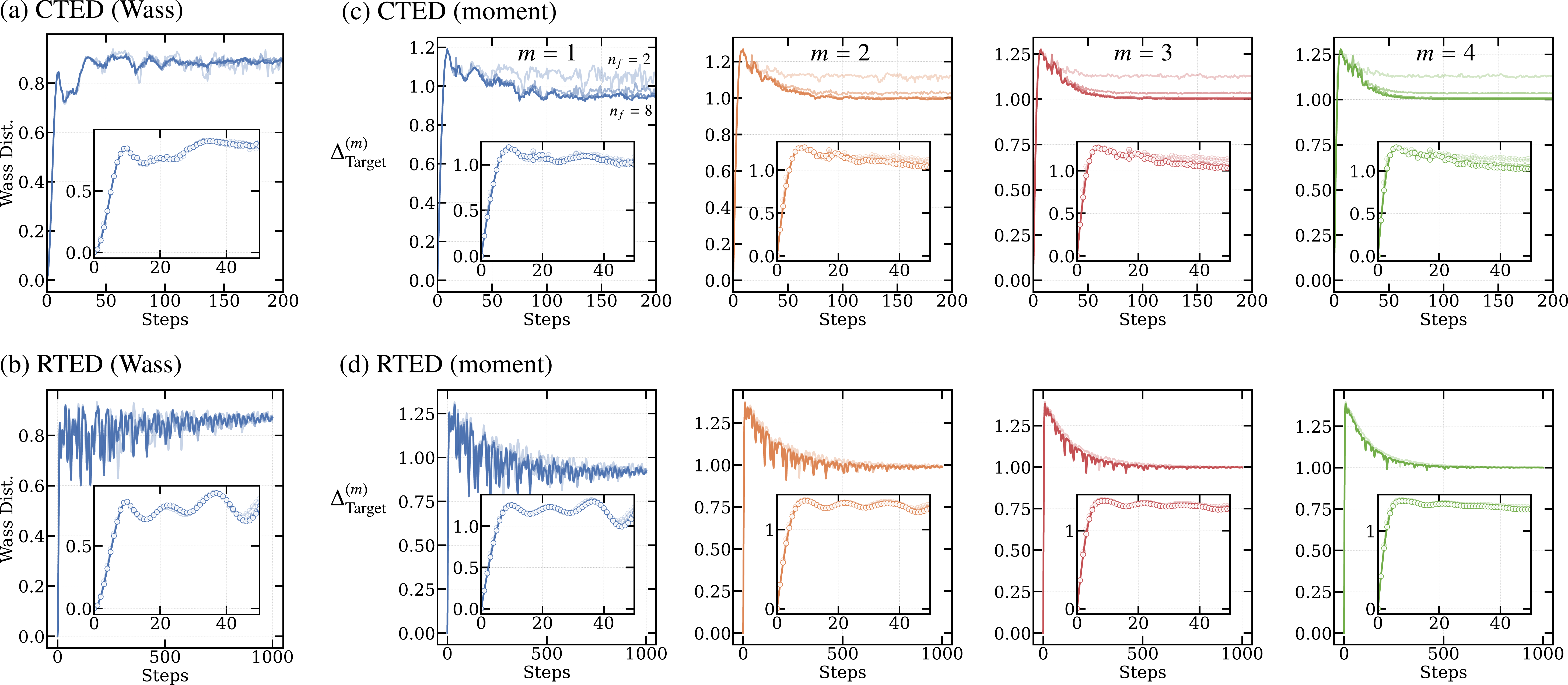}
\caption{1-Wasserstein distance and normalized moment distance $\Delta_{\textrm{Haar}}^{(m)}(k)$ between the diffused ensemble and the target multi-cluster ensemble versus diffusion step $k$, for CTED and RTED. Results are shown for different complement sizes $n_f \in \{2, 4, 6, 8\}$ and moment orders $m$. Inset plots display values at diffusion step $k\leq 50$.}
\label{fig:evol:target}
\end{figure*}

We conduct numerical experiments to generate multi-cluster and circular ensembles of $n_d$-qubit quantum states ($n_d=4$).
The multi-cluster ensemble is a mixture of pure states centered around distinct clusters, modeling multimodal quantum data relevant to fields like quantum chemistry and quantum error correction. 
The target distribution is a mixture of three clusters: 40\% from cluster 1, 40\% from cluster 2, and 20\% from cluster 3, each consisting of \(n_d\)-qubit pure states. Cluster 1 centers on $|0\rangle^{\otimes n_d}$, cluster 2 on $|1\rangle^{\otimes n_d}$, and cluster 3 on the GHZ state $\frac{1}{\sqrt{2}} \left( |0\rangle^{\otimes n_d} + |1\rangle^{\otimes n_d} \right)$.
To mimic imperfections in quantum hardware, we incorporated noise into each cluster by applying random single-qubit rotations, with rotation angles sampled from a Gaussian distribution \(\mathcal{N}(0, \sigma^2)\) where \(\sigma = 0.05\).
For the circular ensemble, the quantum state has the form of $\ket{\psi(\beta)} = \cos(\beta/2)\ket{0...0}+\sin(\beta/2)\ket{1...1}$ of $n_d$ qubits, where $\beta \in [0, 2\pi)$.

We investigate the evolved dynamics of the diffused ensembles derived from CTED and RTED.
Given the classically-enhanced projected ensemble $\cE^{(k)} = \{q(\bx_j^{(k)})p_\bx(\bz_j^{(k)}); \ket{\psi_\bx(\bz_j^{(k)})}\}_j$, we can define the $m$-th moment operator for this ensemble as
\begin{align}
    \rho_k^{(m)} = \sum_j q(\bx_j^{(k)})p_\bx(\bz_j^{(k)}) \left( \ket{\psi_\bx(\bz_j^{(k)})}\bra{\psi_\bx(\bz_j^{(k)})} \right)^{\otimes m}.
\end{align}
We compute the normalized Hilbert-Schmidt distance to the $m$-th moment of the Haar ensemble and the target ensemble:
\begin{align}\label{eqn:norm:HS}
    &\Delta_{\textrm{Haar}}^{(m)}(k) = \dfrac{\| \rho_k^{(m)} - \rho^{(m)}_{\textrm{Haar}} \|_2}{\| \rho^{(m)}_{\textrm{Haar}} \|_2}, \\
    &\Delta_{\textrm{Target}}^{(m)}(k) = \dfrac{\| \rho_k^{(m)} - \rho^{(m)}_{\textrm{Target}} \|_2}{\| \rho^{(m)}_{\textrm{Target}} \|_2}. 
\end{align}

Figures~\ref{fig:evol} 
and~\ref{fig:evol:target} characterize the forward diffusion dynamics induced by chaotic time evolution, quantified by the 1-Wasserstein distance and the normalized moment distances.
Figure~\ref{fig:evol} reports the distance between the diffused ensemble and the Haar-random ensemble for both CTED and RTED, while varying the number of diffusion steps and the complement-system size $n_f \in \{2, 4, 6, 8\}$.
For CTED, all metrics exhibit a rapid initial decay followed by saturation at finite values that depend systematically on $n_f$ and the moment order $m$. Increasing $n_f$ lowers the saturation plateau, indicating that deviations from Haar randomness are dominated by finite-size effects associated with the complement system. Higher-order moments saturate at larger residual values, reflecting the persistence of non-Haar correlations at finer statistical resolution.
RTED shows slower decay with increasing diffusion steps but continues to approach smaller distances at long times. This behavior arises because at the $k$-th diffusion step, RTED repeatedly applies the same short-time unitary evolution while coupling the main system to independently initialized complement subsystems. The effective environment size grows proportionally to $kn_f$, progressively suppressing finite-size effects.

To place these saturation values on an absolute scale, we compare them against a
reference baseline (dashed lines in Fig.~\ref{fig:evol}) given by the residual distance
between two independent finite ensembles of $N$ Haar-random states. This baseline is the
smallest distance attainable at finite sample size and therefore quantifies the fundamental
limit of randomness set by the sampling constraints. 
For the normalized Hilbert--Schmidt moment distance in Eq.~\eqref{eqn:norm:HS}, it takes the closed form
\begin{align}
  \Delta^{(m)}_{\mathrm{floor}}
  = \sqrt{\frac{2\,(1-1/D_m)}{\,1+(N-1)/D_m\,}}
  \;\xrightarrow{\,N\gg D_m\,}\;
  \sqrt{\frac{2D_m}{N}},
  \label{eq:hs_floor}
\end{align}
where $D_m=\binom{D+m-1}{m}$ is the dimension of the symmetric subspace of $m$ copies of the
$D=2^{n_d}$-dimensional Hilbert space. 
For the fidelity-based $1$-Wasserstein distance, an analogous estimate gives
\begin{align}
  W_{\mathrm{floor}}\approx N^{-1/(D-1)}.
  \label{eq:wass_floor}
\end{align}

We now present the derivation for Eqs.~\eqref{eq:hs_floor}\eqref{eq:wass_floor}. 
For two ensembles of pure states $A=\{|a_i\rangle\}_{i=1}^{N_A}$ and
$B=\{|b_j\rangle\}_{j=1}^{N_B}$, let $\hat{\rho}^{(m)}_A=\frac{1}{N_A}\sum_i
(|a_i\rangle\langle a_i|)^{\otimes m}$ denote the empirical $m$-th moment operator, and
define the frame-potential estimator
\begin{align}
  \hat{\mathcal{F}}^{(m)}(A,B)=\mathrm{Tr}\big[\hat{\rho}^{(m)}_A\hat{\rho}^{(m)}_B\big]
  =\frac{1}{N_A N_B}\sum_{i,j}|\langle a_i|b_j\rangle|^{2m}.
\end{align}

The normalized moment distance is then
$\Delta_A^{(m)}=\|\hat{\rho}^{(m)}_A-\hat{\rho}^{(m)}_B\|_2/\|\hat{\rho}^{(m)}_A\|_2$, where $\|\hat{\rho}^{(m)}_A-\hat{\rho}^{(m)}_B\|_2^2
=\hat{\mathcal{F}}(A,A)+\hat{\mathcal{F}}(B,B)-2\hat{\mathcal{F}}(A,B)$. For the baseline, we consider $A$ and $B$ as two independent
ensembles of $N$ Haar-random states.
Let $\Pi_{\mathrm{sym}}$ be the projector onto the symmetric subspace of $(\mathbb{C}^D)^{\otimes m}$
(the states invariant under permuting the $m$ copies), with
$D_m\equiv\mathrm{Tr}\,\Pi_{\mathrm{sym}}=\binom{D+m-1}{m}$ and $D=2^{n_d}$. Using the Haar identity
$\mathbb{E}_{\psi}(|\psi\rangle\langle\psi|)^{\otimes m}=\Pi_{\mathrm{sym}}/D_m$ and the fact that
any product state satisfies $\Pi_{\mathrm{sym}}|\phi^{\otimes m}\rangle=|\phi^{\otimes m}\rangle$,
\begin{align}
  \mathbb{E}_{\psi,\phi}\,|\langle\psi|\phi\rangle|^{2m}
  =\big\langle\phi^{\otimes m}\big|\tfrac{\Pi_{\mathrm{sym}}}{D_m}\big|\phi^{\otimes m}\big\rangle
  =\frac{1}{D_m}.
  \label{eq:haar_overlap}
\end{align}

The frame potentials are evaluated on finite random draws and therefore fluctuate from sample to sample. 
The baseline we report is the typical floor, so we replace each estimator by its Haar expectation. This is justified because each $\hat{\mathcal{F}}$ is an average of $N^2$ bounded terms and concentrates about its mean as $N\to\infty$. 
Separating the $N$ diagonal terms ($|\langle\psi_i|\psi_i\rangle|^{2m}=1$) from the off-diagonal ones gives
$\mathbb{E}[\hat{\mathcal{F}}(A,A)]=\tfrac{1}{N}+(1-\tfrac{1}{N})\tfrac{1}{D_m}$ and
$\mathbb{E}[\hat{\mathcal{F}}(A,B)]=\tfrac{1}{D_m}$. 
The expected squared numerator is thus $\tfrac{2}{N}(1-\tfrac{1}{D_m})$, and normalizing by $\mathbb{E}[\hat{\mathcal{F}}(A,A)]$ yields Eq.~\eqref{eq:hs_floor}.

% For the fidelity-based cost $C_{ij}=1-|\langle a_i|b_j\rangle|^2$ with uniform marginals, the
% optimal transport reduces to a linear assignment,
% $W=\tfrac{1}{N}\min_{\sigma}\sum_i C_{i\sigma(i)}$, so the floor is governed by the best match
% of each reference state. Let $f$ denote a value of the pairwise fidelity
% $|\langle a|b\rangle|^2$. For independent Haar states, its survival function is
% $P(|\langle a|b\rangle|^2>f)=(1-f)^{D-1}$. The largest fidelity $f_{\max}$ attained among the
% $\sim\!N$ candidate matches satisfies $(1-f_{\max})^{D-1}\sim 1/N$, giving the minimized cost
% $(1-f_{\max})\sim N^{-1/(D-1)}$ per pair and hence giving
% Eq.~(\ref{eq:wass_floor}).

For the fidelity-based cost $C_{ij}=1-|\langle a_i|b_j\rangle|^2$ with uniform marginals, the optimal transport reduces to a linear assignment, $W=\tfrac{1}{N}\min_{\sigma}\sum_i C_{i\sigma(i)}$, so the floor is governed by the best match of each reference state. 
Consider one reference state $|a\rangle$ and its $N$ candidate matches $\{|b_j\rangle\}$, and let $f$ denote a value of the pairwise fidelity
$F=|\langle a|b\rangle|^2$. 
For independent Haar states, the fidelity has density $(D-1)(1-f)^{D-2}$, so its survival function is $P(F>f)=(1-f)^{D-1}$.
The assignment pairs $|a\rangle$ with the candidate of largest fidelity, $F_{\max}=\max_{j}F_j$. To estimate the typical value $f_{\max}$ of $F_{\max}$, note that the expected number of candidates whose fidelity exceeds a level $f$ is
$\mathbb{E}\,\#\{j:F_j>f\}=N\,P(F>f)=N(1-f)^{D-1}$: for small $f$ this count is large (many
candidates clear the bar, so the maximum lies above $f$), while for large $f$ it is
$\ll 1$ (no candidate reaches it, so the maximum lies below $f$). 
The maximum therefore concentrates at the crossover level where the expected count is of order one,
\begin{equation}
  N\,P(F>f_{\max})\sim 1
  \Longleftrightarrow
  1-f_{\max}\sim N^{-1/(D-1)}.
  \label{eq:extreme_value}
\end{equation}
Since the minimized cost per pair is $C_{\min}=1-F_{\max}=1-f_{\max}$, the Wasserstein baseline is given by Eq.~\eqref{eq:wass_floor}.

These baselines depend only on $N$, $m$, and $D$, and are independent of the complement size $n_f$.
The saturation plateaus in Fig.~\ref{fig:evol} lie systematically above this baseline.
For our setting in Fig.~\ref{fig:evol} ($n_d=4$, $D=16$, $N=1000$), Eq.~(\ref{eq:hs_floor}) gives $\Delta^{(m)}_{\mathrm{floor}}=\{0.17,\,0.49,\,0.95,\,1.26\}$ for $m=\{1,2,3,4\}$, and Eq.~(\ref{eq:wass_floor}) gives $W_{\mathrm{floor}}\approx 0.63$.
This confirms that the ensembles do not converge to the Haar measure but instead approach a universal finite-temperature ensemble consistent with chaotic dynamics and conserved quantities. 
The gap narrows as $n_f$ grows, as the finite-temperature ensemble more closely reproduces the Haar statistics.
%he results do not imply the convergence to a Haar-random ensemble. 
%Instead, both approach a universal finite-temperature ensemble consistent with chaotic dynamics and conserved quantities.

Figure~\ref{fig:evol:target} compares the diffused ensemble directly to the target multi-cluster distribution. For both CTED and RTED, the distance to the target ensemble increases during the diffusion process and subsequently saturates at a finite value, indicating that the forward diffusion drives the ensemble away from the target distribution toward a stationary noisy ensemble.
Each step mixes the main system $\cM$ via chaotic dynamics and then discards information by measuring the complement system $\cF$. Repeating this drives the moment statistics toward a steady regime determined by the scrambling dynamics and the complement size $n_f$.
The approach to that steady regime is not governed by a single timescale: different components of the moment operator relax at different rates. When we measure a single scalar distance to the target, the fast components can move away from the target at early steps, and only later do slower components pull back toward the steady plateau.
CTED exhibits smoother convergence, whereas RTED shows larger temporal fluctuations due to repeated reinitialization of $\cF$, while reaching comparable asymptotic distances.

We consider the forward process with $K=50$ steps, $n_f=2$, and $n_a=2$ in the backward process.
We train each denoising operation $V_k$ for 1000 epochs on 1000 samples with a mini-batch size of 100, minimizing the 1-Wasserstein distance.
Theoretically, convergence of the projected ensemble toward an $m$-design requires $n_f=\Omega(m\,n_d)$, so larger $n_f$ lowers the
residual deviation (consistent with Fig.~\ref{fig:evol}, where increasing $n_f$ lowers the plateau) at the cost of more qubits and longer evolution. 
In practice, the requirement is milder than worst-case scaling: $n_f$ need only be large enough that the forward channel sufficiently ``forgets'' the input while the backward model remains trainable. 
We therefore recommend a simple operational criterion: increase $n_f$ until generative (backward) performance stops improving appreciably, rather than fine-tuning hyperparameters. 
For the datasets considered here, this occurs already at the modest value $n_f=2$ adopted above.
We further note that RTED relaxes the demand on $n_f$, because the effective environment grows as $\sim k\,n_f$ across repeated $k$ steps, so a smaller $n_f$ suffices to suppress finite-size effects.

\begin{figure*}
\centerline{\includegraphics[width=1.0\linewidth]{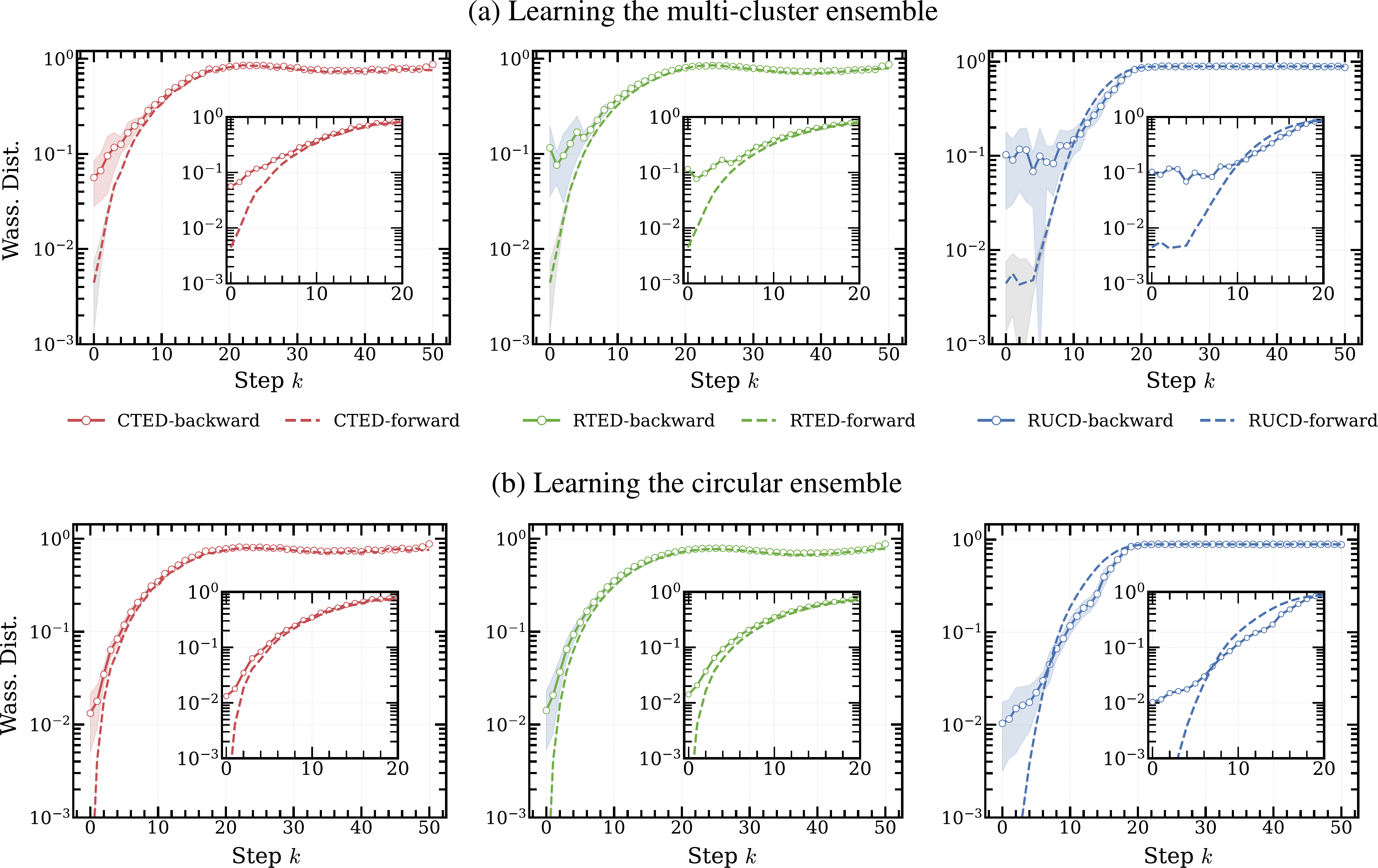}}
\caption{1-Wasserstein distances between the generated ensembles and the true ensemble over backward steps,  sampled from (a) multi-cluster and (b) circular datasets of quantum states, for CTED (left), RTED (center), and RUCD (right). 
The backward process (solid lines with circle markers) is compared to the forward diffusion (dotted lines). 
Lines and shaded areas represent the mean and standard deviation over ten trials. Inset plots display values at diffusion step $k\leq 20$.}
\label{fig:compare:QDM}
\end{figure*}

Figure~\ref{fig:compare:QDM} compares the backward generative performance of CTED, RTED, and RUCD using the 1-Wasserstein distance between the generated ensemble $\tcS^{\prime}_k$ and the true ensemble $\cS_0$, for both multi-cluster and circular quantum datasets. In general, the final distances are comparable for all methods. In particular, for CTED and RTED, the backward process closely mirrors the corresponding forward diffusion, showing smooth, monotonic convergence with relatively low variance across trials. This indicates that the chaotic diffusion dynamics generate structured noise that can be effectively inverted by the backward process. In contrast, RUCD exhibits significantly larger fluctuations and a pronounced mismatch between forward and backward curves, particularly at early steps, reflecting the difficulty of inverting near-Haar random diffusion. 

\subsection{Learning Quantum Chemistry Distribution}

\begin{figure*}
    \centerline{\includegraphics[width=0.9\linewidth]{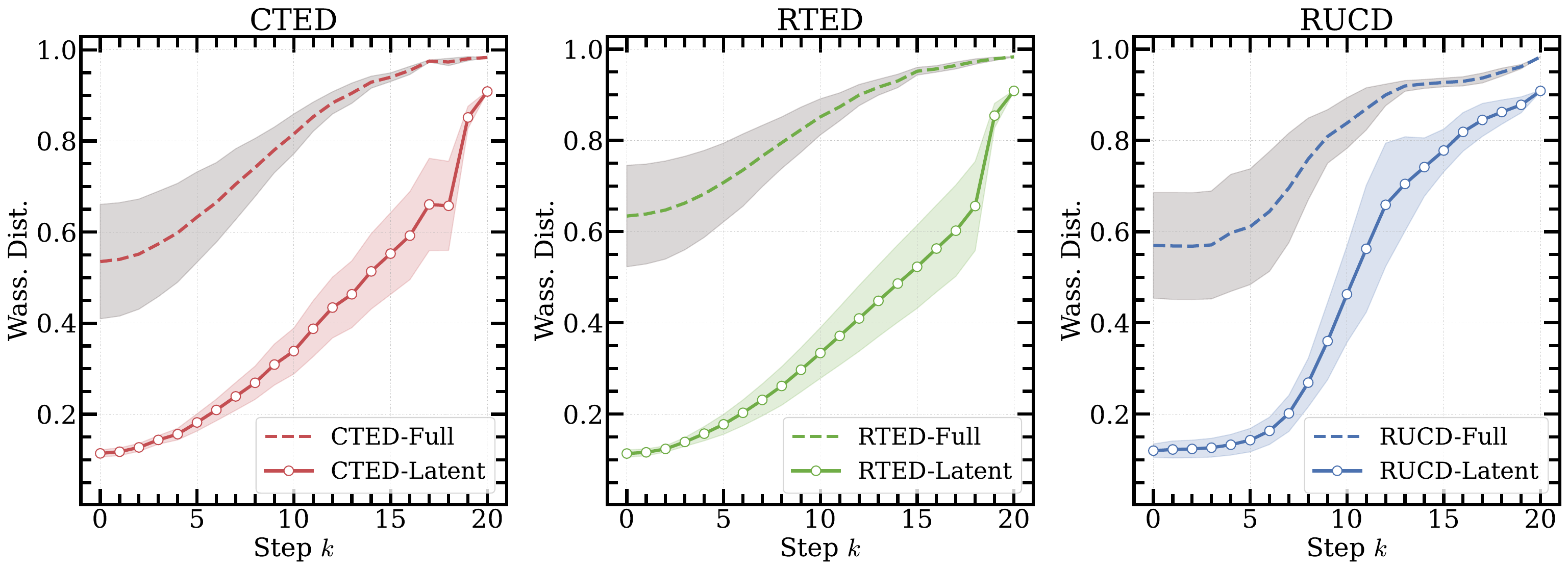}}
\caption{1-Wasserstein distance between the generated ensemble and the true ensemble of the QM9 subset dataset as a function of diffusion step $k$, comparing backward diffusion performed in the full Hilbert space (“Full”, dashed) and in the latent space learned by a quantum autoencoder (“Latent”, solid). Results are shown for CTED (left), RTED (center), and RUCD (right). %Latent-space diffusion systematically slows the growth of the distance and reduces fluctuations, indicating more controlled noise injection compared to full-space diffusion.
Shaded regions denote the standard deviation over ten trials.}
\label{fig:compare:qm9}
\end{figure*}

We apply our framework to the QM9 dataset~\cite{ramakrishnan:2014:QM9}, a standard benchmark in computational chemistry. QM9 comprises approximately $1.34 \times 10^5$ equilibrium geometries of small organic molecules containing up to nine heavy atoms (C, N, O, and F), along with additional hydrogen atoms, for a maximum of 29 atoms per molecule. Each entry includes molecular properties and three-dimensional structural information, making the dataset widely used for tasks such as molecular property prediction and 3D structure generation~\cite{wu:qvae-mole:2024}.
Because of current computational limitations, processing the full dataset is not feasible. We therefore restrict our study to molecules with exactly eight heavy atoms and two distinct ring systems, resulting in a more structurally homogeneous subset of 4236 molecules. The 3D structures of these molecules are encoded into 7-qubit quantum states following the encoding scheme in Ref.~\cite{wu:qvae-mole:2024}.

In our numerical experiments, we set $n_d = 7$, $n_f = 3$, and use $K = 20$ diffusion steps. We implement each backward denoising unitary $V_k$ with a hardware-efficient ansatz with $L=10$ layers. We train each $V_k$ for 1000 epochs on 2000 training samples with a mini-batch size of 100, minimizing the 1-Wasserstein distance.
For this nontrivial distribution in a high-dimensional Hilbert space, however, accurate backward denoising can require deeper circuits and longer training, making direct learning in the full 7-qubit space challenging.

To mitigate this, we adopt a latent-space strategy~\cite{Naipunnya:2025:latent} based on a quantum autoencoder (QAE)~\cite{Romero:2017:QAE}. 
This idea is also aligned with prior quantum generative modeling works that train in an extracted latent feature space to improve learnability and output quality in image-generating tasks~\cite{chang:2024:latengQuantumGAN}.
Specifically, we first train a QAE to compress the 7-qubit states into a $n_l=4$-qubit latent representation, and then perform diffusion and denoising in the latent space. The QAE encoder is implemented as a depth-20 hardware-efficient circuit consisting of single-qubit $RY$ parameterized rotations followed by a ring of CNOT gates repeated across layers. It is trained to decouple the remaining ``trash" qubits by maximizing their overlap with the $\ket{0\ldots 0}$ state. We optimize the QAE parameters using the Adam optimizer with a learning rate of 0.001 for 2000 epochs. After training, the decoder (implemented as the approximate inverse circuit) maps generated latent states back to the original 7-qubit space, enabling sampling in the full molecular state representation while reducing the learning burden of the diffusion model.
In the latent space, we perform forward and backward diffusion with $n_d=4$, $n_f=2$ qubits, and use $K=20$ diffusion steps. Other training conditions match those used for training in the full 7-qubit space.

Figure~\ref{fig:compare:qm9} compares the performance of the backward diffusion (denoising) process when it is learned directly in the full Hilbert space (“Full”) and when it is learned in the latent space obtained from a quantum autoencoder (“Latent”), followed by decoding back to the original space.
For all three diffusion schemes (CTED, RTED, and RUCD), backward denoising learned in the latent space consistently achieves lower Wasserstein distances and reduced fluctuations compared to denoising learned directly in the full Hilbert space. This indicates that backward maps trained in the latent representation can invert the diffusion process more accurately, particularly at intermediate and late backward steps where the input states are more strongly corrupted. The improvement is especially evident in the smoother growth of the Wasserstein distance and the smaller variance across trials, demonstrating enhanced stability and robustness of the learned denoising dynamics.

This behavior can be understood by viewing the quantum autoencoder as an information bottleneck that reshapes the inverse problem faced by the backward diffusion. By compressing the original 7-qubit molecular states into a lower-dimensional latent space, the autoencoder isolates the dominant degrees of freedom that characterize the data distribution while discarding directions that contribute weakly to reconstruction but strongly to noise amplification. As a result, the backward learned in the latent space operates on a reduced manifold where the inverse mapping from noisy states to cleaner ones is better conditioned. After decoding, the reconstructed full-space states therefore remain closer to the target ensemble than those obtained by backward denoising trained directly in the full Hilbert space.

% The effect of latent-space denoising is observed across all diffusion mechanisms but is most pronounced for RUCD. In the full space, RUCD backward diffusion exhibits comparatively larger errors and fluctuations, reflecting the difficulty of inverting strongly scrambling, near–Haar-random dynamics. When trained in the latent space, however, RUCD backward denoising becomes significantly more stable and accurate, indicating that the autoencoder suppresses highly randomizing directions that are detrimental to inversion. For CTED and RTED, which already generate more structured diffused ensembles, latent-space denoising further improves stability and reduces variance, leading to consistently better backward reconstruction. Overall, Fig. 6 demonstrates that learning the backward diffusion in a learned latent representation substantially improves the accuracy and conditioning of the denoising process, especially for complex, high-dimensional quantum data distributions.

\subsection{Resource Comparison}
To clarify the resource advantage of CTED and RTED, we compare its implementation cost with RUCD in terms of entangling-gate complexity and circuit depth. In our implementation with RUCD, each diffusion layer contains an all-to-all entangling block, which requires $O(n_d^2)$ two-qubit gates. Over $K$ diffusion steps, the total entangling gate count scales as $O(Kn_d^2)$ per sample, leading to significant compilation and calibration overhead.

Beyond gate-count scaling, RUCD also incurs a substantial compilation overhead. Each diffusion step requires implementing distinct random unitary circuits, which must be decomposed into the native gate set of the target hardware. This results in a large number of circuit instances with different parameterizations, increasing the burden of circuit synthesis, calibration, and control. On near-term quantum devices, such repeated compilation and calibration can dominate the experimental cost, particularly for deep or highly entangling circuits.

In contrast, CTED and RTED implement diffusion through continuous-time evolution under a fixed chaotic Hamiltonian. On analog platforms such as Rydberg arrays or trapped-ion systems, this Hamiltonian is realized via global, time-independent control fields. Each diffusion step corresponds to physical time evolution rather than a compiled gate sequence, and the required control resources do not scale with $n_d^2$ in the same manner.
The primary cost is total evolution time $T=K\Delta t$, rather than circuit depth.
Concretely, the diffusion process is realized by varying only the evolution time or repeating the same physical dynamics, without recompiling into different gate sequences.

The two schemes, CTED and RTED, trade different hardware resources. CTED requires coherent evolution over the cumulative time $k\Delta t$, but only a single complement measurement and reset per sample. RTED requires only short, fixed-duration evolutions, but a fast, high-fidelity measurement-and-reset of $\mathcal{F}$ at every step. CTED is therefore favored when coherence times are long, but the reset is slow. RTED is favored when reset and mid-circuit measurements are fast and accurate.

The advantages of our method extend beyond quantum analog platforms. Even on gate-based quantum computers, reusing a fixed Hamiltonian evolution is more efficient than recompiling distinct circuits at each step. 
The cumulative CTED propagator $e^{-iHk\Delta t}$ and the repeated RTED propagator are realized by repeating the same fixed Trotter step $e^{-iH\Delta t}$, with only the number of repetitions changing, rather than synthesizing a new random circuit at each step.
For example, consider implementing $e^{-iH\Delta t}$ in the recently proposed space-time efficient analog rotation quantum computing (STAR)  architecture~\cite{akahoshi:prxquantum:STAR,toshio:prx:star:2025}, where Clifford operations are fault-tolerant and arbitrary Pauli rotations are implemented via repeat-until-success (RUS) analog rotations.
For the mixed-field Ising Hamiltonian $H$ defined on $n$ qubits, a first-order Trotter step decomposes into commuting layers of single-qubit and nearest-neighbor two-qubit rotations, requiring $3n-1$ Pauli rotation gates and $O(n)$ Clifford gates. In the STAR model, each Pauli rotation takes on average a constant number of clock cycles, and commuting rotations can be parallelized, yielding a runtime per step that scales as $O(n)$ in serial execution and $O(\log n)$  under parallel execution. In CTED and RTED, this fixed evolution operator is reused across all data instances, so the compilation and control overhead is independent of the dataset size.
However, on gate-based hardware, first-order Trotterization carries an error of $O(\Delta t^2)$ per step. Controlling this error to a fixed accuracy may require higher-order product formulas or a smaller time step, increasing the gate count and the associated noise on NISQ devices. This overhead is absent for native analog Hamiltonian evolution.

These distinctions become especially important on analog hardware, where global Hamiltonian evolution is native but synthesizing deep random circuits with all-to-all entangling gates is experimentally demanding. 
The reuse of a single control Hamiltonian significantly reduces compilation overhead and improves experimental feasibility, especially on analog platforms where programmable gate synthesis is limited.
Hence, chaotic diffusion provides a hardware-efficient alternative to circuit-based scrambling while maintaining comparable generative performance

\subsection{Noise Robustness}

We investigate the robustness of CTED, RTED, and RUCD to explicit stochastic noise and quantify the degradation in the 1-Wasserstein distance between the generated and target ensembles.
We emphasize that these noise studies are not intended as a one-to-one cross-method comparison, since the dominant error mechanisms differ qualitatively. Here, RUCD is primarily affected by discrete gate and control errors that accumulate with the number of applied pulses. In contrast, CTED and RTED are primarily affected by time-continuous decoherence during analog Hamiltonian evolution and by measurement and reset imperfections.
Accordingly, we perform within-method sensitivity analyses using noise models that reflect the most relevant physical mechanism for each setting.

For RUCD, we inject stochastic single-qubit Pauli errors after each single-qubit rotation gate. After every application of $ e^{-ig Y /2}$ and $e^{-ig Z /2}$ in the scrambling circuit, we apply a random Pauli operator $P\in \{X, Y, Z\}$ with probability $p_1$, chosen uniformly among the three Paulis. When no error occurs, the circuit proceeds ideally. This corresponds to a per-gate Pauli-depolarizing channel acting after the unitary $U$ at each noisy gate location:
\begin{align}
    \cN(\rho) = (1-p_1)U\rho U^\dagger + \dfrac{p_1}{3}\sum_{P\in \{X, Y, Z\}} PU\rho U^\dagger P.
\end{align}
Since such errors accumulate with circuit depth, the effective corruption strength scales approximately with the number of single-qubit gate locations $N_{1q}$. For small $p_1$, the expected number of injected Pauli errors approximates $p_1N_{1q}$, and the overall deviation increases roughly linearly in $p_1$ at fixed depth.

For CTED and RTED, we do not assume the absence of mid-circuit errors. Rather, we model noise as decoherence acting continuously during the analog evolution. 
Concretely, we consider standard Markovian pure dephasing implemented at the state-vector trajectory level~\cite{brun:2000:trajectories}. After coherent evolution $U(t) = e^{-iHt}$, each qubit undergoes an independent stochastic $Z$ flip with probability $p_{\phi}(t) = \dfrac{1-e^{-\gamma_{\phi}t}}{2}$,
which corresponds to an unraveling of phase-damping noise (the chaotic evolution is unitary when $\gamma_{\phi}=0$).
This dephasing noise model reflects the dominant experimental imperfections on the analog platforms~\cite{levine:2018:dephasing,Zhang:2020:ion:dephasing,trapped-ion:2020:dephasing}.
We define $p_2 = p_{\phi}(\Delta t)$ as the dephasing strength after one diffusion step duration in the RTED. Then $p_{\phi}(k\Delta t) = \dfrac{1 - (1-2p_2)^k}{2}$ is the noise probability applied after $k$th step in CTED.
Note that $p_1$ (in RUCD) and $p_2$ (in CTED and RTED) parameterize different physical mechanisms (per-gate control faults versus time-integrated decoherence) and are therefore not directly comparable in hardware realism or effective channel strength.

\begin{figure*}
    \centerline{\includegraphics[width=0.9\linewidth]{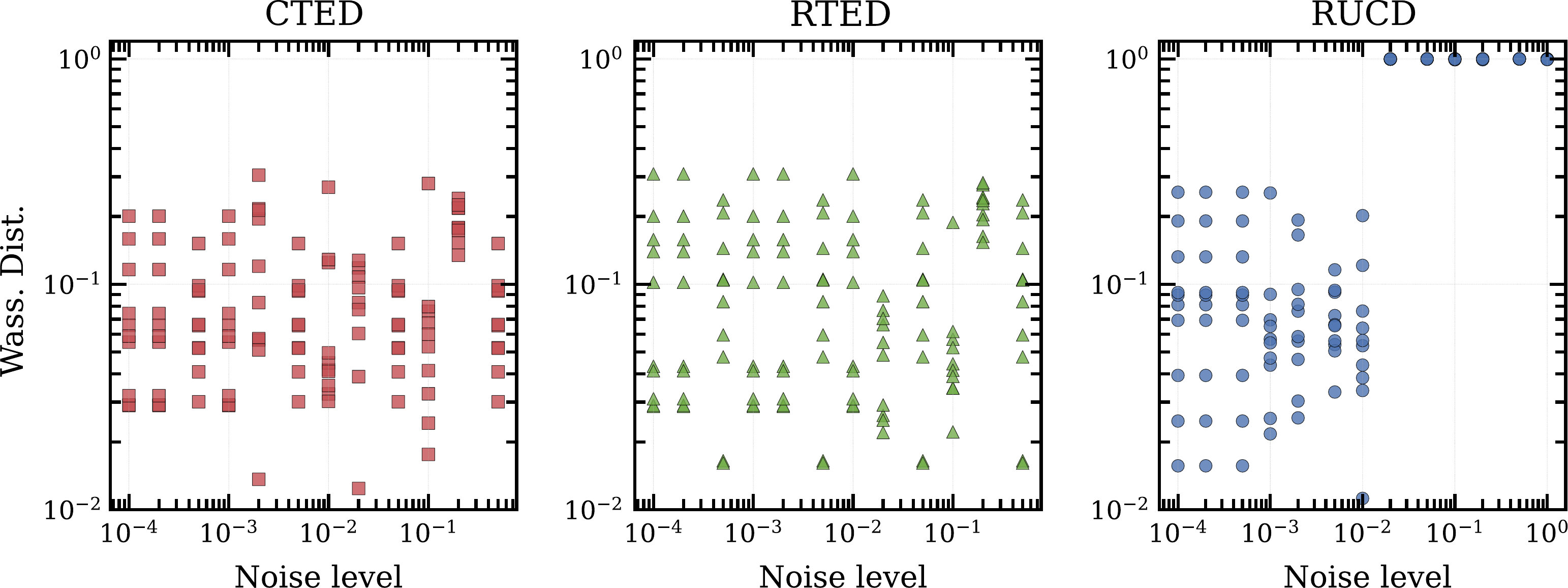}}
\caption{Wasserstein distance between the generated ensemble and the target ensemble as a function of the noise level for CTED (left), RTED (center), and RUCD (right). Each marker corresponds to an independent run. CTED and RTED exhibit gradual degradation as noise increases, whereas RUCD shows a sharp loss of performance at higher noise levels. Here, we consider different physical noise mechanisms for the gate-based ($p_1$ with RUCD) and analog ($p_2$ with CTED and RTED) models.}
\label{fig:compare:noise}
\end{figure*}

Figure~\ref{fig:compare:noise} summarizes the noise sensitivity of CTED, RTED, and RUCD as the noise levels $p_2$ and $p_1$ increase for the task of learning the multi-cluster distribution of quantum data.
Here, we consider the setting in which samples from the original (clean) target ensemble are not available during training. 
Instead, the model is trained only on diffused samples, which serve as the noisy training data. CTED and RTED remain stable across the noise range, with Wasserstein distances staying bounded.
We notice that the horizontal axes in Fig.~\ref{fig:compare:noise} parameterize different physical
noise mechanisms for the gate-based ($p_1$ with RUCD) and analog ($p_2$ with CTED and RTED) models. The figure should be read as a within-method sensitivity analysis, not as a one-to-one comparison of noise
strengths across architectures.
%In particular, the performance deterioration is smooth, indicating that the learned backward denoising process remains effective even when the forward scrambling is imperfect.

RUCD exhibits a much sharper breakdown. Beyond a moderate noise level, many instances rapidly jump to Wasserstein distances of order unity, indicating a failure to recover the target ensemble. This behavior is consistent with the structure of the injected noise. Because RUCD relies on explicit circuit-based scrambling composed of many single-qubit rotations, it contains a large number $N_{1q}$ of noisy gate locations. At high $p_1$, the accumulated corruption grows quickly with circuit depth, effectively driving the diffused states toward highly randomized distributions at early forward steps. Consequently, the learned backward map cannot recover the target, as the forward channel has already erased the necessary information.

Unlike the gate-based implementation of RUCD, where errors accumulate across many noisy operations, analog execution of CTED and RTED involves fewer locations where this noise model applies, resulting in less accumulated error. CTED and RTED generate the ensemble by measuring the complement system $\cF$ and retaining the post-measurement state on the main system $\cM$. These measurements collapse and discard the complement subsystem, preventing errors in $\cF$ from coherently propagating across multiple steps. 
Note that CTED and RTED are still sensitive to time-integrated decoherence affecting $\cM$ during evolution.
%This process also converts certain coherent error effects into classical randomness in the ensemble weights, which is often less detrimental to distribution learning than coherent noise. Therefore, the projected ensemble generation mechanism naturally suppresses error propagation pathways that would otherwise be present in deep coherent circuits.

We provide a mathematical explanation for the intrinsic robustness of CTED and RTED under coherent errors in the complement system $\cF$. 
Each diffusion step in chaotic diffusion schemes consists of (i) coupling the main system $\cM$ to a complement $\cF$ via a joint Hamiltonian evolution, followed by (ii) a projective measurement on $\cF$ and discarding $\cF$. 
%This defines a \textit{quantum instrument} acting on $\cM$, which converts certain coherent errors on $\cF$ into classical randomness in the measurement record.
% Recall that each diffusion step prepares a joint pure state $\ket{\Phi}$ on $\cM+\cF$,
% followed by a projective measurement of $\cF$ in the computational basis $\{ \ket{\bz_\cF} \}_\bz$. The resulting (unnormalized) post-measurement state on $\cM$ is
% \begin{equation}
% \ket{\tilde{\Phi}_\cM(\bz_\cF)} = (\bbI_\cM \otimes \bra{\bz_\cF})\,|\Phi\rangle,
% \qquad
% p(z)=\|\,\ket{\tilde{\Phi}_\cM(\bz_\cF)}\,\|^2,
% \qquad
% |\Phi_\cM(\bz_\cF)\rangle = \frac{|\tilde{\Phi}_{\cM}(\bz_\cF)\rangle}{\sqrt{p(\bz_\cF)}}.
% \label{eq:pe_pure}
% \end{equation}
This measure-and-discard structure has two stabilizing consequences. First, once $\cF$ is measured and discarded, any noise acting on $\cF$ cannot coherently propagate through $\cF$ across diffusion steps (especially in RTED where $\cF$ is re-prepared at each step). Second, errors on $\cF$ immediately before measurement effectively modify the measurement rule on $\cF$ rather than inducing a coherent, accumulating miscalibration on $\cM$. We formalize this observation by showing that arbitrary noise on $\cF$ prior to the measurement is equivalent to measuring $\cF$ with a noisy POVM, and that Pauli noise on measured qubits leads only to outcome relabeling, leaving permutation-invariant ensemble statistics unchanged.

\begin{lemma}[Noise becomes POVM]
\label{lem:noise_povm}
Let $\cN_\cF$ be an arbitrary CPTP noise channel acting on $\cF$ immediately before the measurement. Then the measurement statistics under $\cN_\cF$ are identical to those obtained by measuring the noiseless state $\ket{\Phi}$ with a POVM $\{E_\bz\}_{\bz}$ on $\cF$ defined by
\begin{equation}
E_\bz := \cN_\cF^\dagger(\ket{\bz}\bra{\bz}),
\qquad E_\bz \succeq 0,
\qquad \sum_\bz E_\bz = \bbI_F,
\label{eq:Ez_def}
\end{equation}
where $ \cN_\cF^\dagger$ is the adjoint (Heisenberg-picture) map. In particular, the noisy outcome probability is
\begin{equation}
\tilde{p}(\bz)
= \langle \Phi | (\bbI_\cM \otimes E_\bz) |\Phi\rangle.
\label{eq:noisy_prob_povm}
\end{equation}
\end{lemma}

\begin{proof}
Let $\rho=|\Phi\rangle\langle\Phi|$. Applying $\cN_\cF$ on $\cF$ gives
$\tilde{\rho}=(\bbI_\cM\otimes \cN_\cF)(\rho)$.
The noisy probability of measuring $z$ is
\begin{equation}
\tilde{p}(\bz) = \tr\!\left[(\bbI_M \otimes \ket{\bz}\bra{\bz})\,\tilde{\rho}\right]
= \tr\!\left[\big(\bbI_\cM \otimes \cN_\cF^\dagger(\ket{\bz}\bra{\bz})\big)\,\rho\right],
\end{equation}
which yields Eq.~\eqref{eq:noisy_prob_povm} with $E_z$ defined in Eq.~\eqref{eq:Ez_def}. Positivity and completeness follow from complete positivity and trace preservation of $\cN_\cF$.
\end{proof}

Lemma~\ref{lem:noise_povm} formalizes the classicalization mechanism: any noise on the measured subsystem can be absorbed into the measurement rule (a noisy POVM), rather than coherently accumulating on $\cM$ as an unknown unitary distortion.

We next specialize to Pauli errors acting on measured qubits in $\cF$ immediately before a computational-basis measurement. Let $q\in \cF$ denote a measured qubit and $\be_q$ the bitstring with a $1$ at position $q$ and $0$ elsewhere.

\begin{corollary}[Pauli relabeling invariance]
\label{cor:pauli_relabel}
Consider a projected-ensemble step in which $\cF$ is measured in the computational basis. If, immediately before measurement, a Pauli operator $P_q\in\{X_q,Y_q,Z_q\}$ acts on a measured qubit $q\in \cF$, then:
(i) $Z_q$ leaves the conditional (unnormalized) post-measurement state on $\cM$ unchanged up to a global phase and hence does not change the projected ensemble;
(ii) $X_q$ and $Y_q$ permute (relabel) the measurement outcome as $\bz \mapsto \bz\oplus \be_q$ (up to a phase for $Y_q$), leaving the set of conditional states invariant up to relabeling.
Consequently, any permutation-invariant ensemble functional is unchanged, including the moment operators.
\end{corollary}

\begin{proof}
Let $|\Phi\rangle$ be the joint pre-measurement state on $\cM+\cF$ and define
$\ket{\tilde{\Phi}_\cM(\bz)} = (\bbI_\cM \otimes \bra{\bz})\ket{\Phi}$.
If $Z_q$ acts before measurement, then $Z_q\ket{\bz} = (-1)^{z_q}\ket{\bz}$, where $z_q$ is the bit value of the qubit $q$. Hence,
$|\tilde{\Phi}^{(Z_q)}_\cM(\bz)\rangle
= (\bbI_\cM\otimes \bra{\bz}) (\bbI_\cM\otimes Z_q)\ket{\Phi}
= (-1)^{z_q} |\tilde{\Phi}_M(z)\rangle,$
which changes only a global phase and therefore leaves both $p(\bz)=\| \ket{\tilde{\Phi}_\cM(\bz)}\|^2$ and the normalized conditional state $|\Phi_\cM(\bz)\rangle$ unchanged.

If $X_q$ acts before measurement, then $X_q\ket{\bz} = \ket{\bz\oplus \be_q}$, giving
$
|\tilde{\Phi}^{(X_q)}_\cM(\bz)\rangle
= (\bbI_\cM\otimes \bra{\bz}) (\bbI_\cM\otimes X_q)\ket{\Phi}
= (\bbI_\cM\otimes \bra{\bz\oplus \be_q})\ket{\Phi}
= \ket{\tilde{\Phi}_\cM(\bz\oplus \be_q)},
$
so the conditional outputs are permuted by $\bz\mapsto \bz\oplus \be_q$.
For $Y_q=iX_qZ_q$, the same permutation holds up to an outcome-dependent phase, which again does not affect $p(\bz)$ or the normalized conditional state.
Finally, moment operators and other permutation-invariant ensemble statistics are unchanged under relabeling of outcomes, since the sums over outcomes are invariant under bijections of the index set. This completes the proof.
\end{proof}

% In RTED, each diffusion step re-prepares the complement $F$ and measures/discards it afterward. Consequently, noise on $F$ at step $k$ cannot coherently affect step $k+1$ through quantum memory in $F$. Mathematically, the evolution of $M$ forms a Markovian stochastic process driven by the instrument $\{\mathcal{M}_z\}_z$, where the trajectory distribution factorizes as $p(z_{1:K})=\prod_{k=1}^K p(z_k\,|\,\rho_M^{(k-1)})$. This ``measure-and-reset'' structure suppresses coherent error build-up pathways that are present in deep coherent circuits, contributing to the observed robustness of chaotic projected-ensemble diffusion.

\section{Summary and Discussion}
We have proposed a quantum diffusion model with two schemes, CTED and RTED, using chaotic Hamiltonian evolution. These schemes match RUCD's accuracy in learning quantum data distributions while eliminating its complex spatio-temporal control, making them ideal for analog quantum systems. 
Our experiments on synthetic multi-cluster and circular datasets demonstrate that CTED and RTED achieve generative accuracy comparable to RUCD. In the QM9 molecular dataset, we further show that combining chaotic diffusion with a quantum autoencoder enables stable learning in a compressed latent space, significantly improving performance in high-dimensional Hilbert spaces. 

Our chaotic diffusion processes are designed to be compatible with analog quantum simulators where dynamics is generated by a time-independent many-body Hamiltonian with minimal spatiotemporal control. In our simulations, we used a mixed-field Ising model on a 1D chain, where the competition between interactions and non-commuting fields places the system in a chaotic (thermalizing) regime away from integrable limits. This structure maps naturally to several leading platforms. In Rydberg atom arrays, effective Ising couplings arise from blockade-induced interactions, while a global Rabi drive and detuning implement the transverse ($X$) and longitudinal ($Z$) fields~\cite{Bernien:2017:rydberg}.
Site-resolved measurement enables direct implementation of the complementary measurement required in the projected ensemble, and local re-initialization can support RTED-style repeated short evolutions. 
In trapped-ion simulators, phonon-mediated interactions generate programmable long-range $Z_iZ_j$ coupling with tunable range, and global drives implement transverse fields~\cite{Britton:2012:trapped}. The complementary measurement and reset are feasible within typical experimental workflows. Cold-atom implementations are most direct in Rydberg-dressed lattice settings (similar to Rydberg arrays), while superexchange-based lattice spin models may require additional engineering to realize the mixed-field Ising form~\cite{Gross:2017:cold:atoms}.

% Connectivity and boundary conditions mainly affect the mixing time and the stationary projected ensemble itself, since our endpoint is generally not Haar at finite energy density.
% This yields a platform-specific schedule that fixes the effective diffusion strength while allowing the underlying connectivity to vary.

Because chaoticity can be imperfect in realistic devices, it is important to understand sensitivity to near-integrable regimes. 
%Our framework does not require maximal scrambling. Rather, it requires that the forward channel sufficiently ``forget" detailed input information while remaining learnable by the backward model. 
In Sec.~\ref{subsec:cluster:circular}, we use the Hamiltonian in Eq.~\eqref{eqn:Ising:Ham} at the operating point $(h^x, h^y) = (0.8090, 0.9045)$, and $J = 1.0$ in the chaotic region. To investigate the sensitivity along the chaotic region, we sweep the fields $(h^x,h^y)$ along two cuts through the operating point (fixed $h^x=0.8090$, varying $h^y$; and fixed $h^y=0.9045$, varying $h^x$), retraining the model with 10 random seeds per point.
To ensure a fair comparison, we perform the per-step scrambling by normalizing the time step as $\Delta t = c/\Lambda(h^x,h^y)$, where $\Lambda(h^x,h^y)$ is the spectral width of the Hamiltonian and $c$ is set so that $\Delta t$ reduces to the value used at the operating point.
We find that the final Wasserstein distance shows no systematic dependence on the fields within the chaotic regime. Therefore, we report the results only at the operating point.
We interpret this insensitivity as evidence that the diffusion framework does not require maximal scrambling. Rather, it requires that the forward channel sufficiently ``forget" detailed input information while remaining learnable by the backward model. 
For the low-complexity target distributions considered in our tasks, this condition holds across the entire chaotic region, so generative performance is governed by the expressivity of the backward model rather than by the details of the forward dynamics, and any robustly chaotic operating point performs comparably.
We expect the influence of the forward dynamics, and hence chaoticity, to become more pronounced for harder, higher-entropy target distributions.

%In practice, this can be assessed by sweeping Hamiltonian parameters and monitoring simple forward indicators such as purity decay or local observable relaxation alongside backward-model stability and sample quality. We expect degradation primarily when dynamics approaches integrable limits, where mixing slows and the projected ensemble becomes strongly structured. Away from these limits, the method should remain robust under realistic connectivity and moderate Hamiltonian imperfections.

Beyond hardware efficiency, chaotic diffusion modifies the structure of the forward process. RUCD is designed to approximate Haar-random ensembles at large diffusion steps, whereas CTED and RTED converge to finite-temperature projected ensembles constrained by conserved quantities. In our numerical experiments, this difference manifests in smoother forward trajectories and reduced fluctuations during backward training for CTED and RTED compared to RUCD. While all three approaches achieve comparable final generative accuracy in the regimes studied, the chaotic diffusion schemes avoid full Haar over-scrambling and may provide a better-conditioned inverse problem in certain parameter regimes. 
However, a fair comparison between these methods requires a scrambling schedule that matches the scrambling scale across protocols. For example, we can match the scrambling level by calibrating each forward process to the same scalar corruption scale.
A systematic theoretical analysis of this potential advantage remains an important direction for future work.

Finally, throughout this work we employed a generic hardware-efficient ansatz for
the backward denoising circuits in order to isolate the effect of the forward diffusion
mechanism from ansatz engineering. Problem-tailored ans\"atze offer a natural route to improving generative accuracy, particularly for high-dimensional data, and a systematic study of the interplay between ansatz design and the latent-space strategy is left for future work.

%\bibliography{main.bib}

%

\end{document}